\documentclass{article}[11pt]
\usepackage{amsfonts}
\usepackage{amsmath}
\usepackage{mathrsfs}
\usepackage{color}
\usepackage{amsmath, amsfonts, amssymb, amsthm, amscd}
\usepackage[table]{xcolor}
\usepackage{multicol}
\usepackage{multirow}

\newcommand{\be}{\begin{eqnarray}}
\newcommand{\ee}{\end{eqnarray}}
\newcommand{\bee}{\begin{eqnarray*}}
\newcommand{\eee}{\end{eqnarray*}}

\newtheorem{theorem}{Theorem}[section]

\newtheorem{lemma}[theorem]{Lemma}

\newtheorem{proposition}[theorem]{Proposition}

\newtheorem{assumption}[theorem]{Assumption}

\begin{document}
\title{Constrained portfolio-consumption strategies with uncertain
parameters and borrowing costs
\thanks{The authors thank the Editor-in-Chief, the Associate Editor, and the two referees for
their valuable comments and suggestions. Yang's work is supported by
NNSF of China (Grant No. 11771158, 11801091). Zhou's work is
supported by Singapore MOE (Ministry of Education's) AcRF grant
R-146-000-219-112 and R-146-000-255-114.}}
\author{
Zhou Yang\thanks{%
School of Mathematical Sciences, South China Normal
 University, Guangzhou 510631, China; email: \texttt{%
yangzhou@scnu.edu.cn}} \and
Gechun Liang\thanks{%
Department of Statistics, University of Warwick, Coventry, CV4 7AL,
U.K.; email: \texttt{g.liang@warwick.ac.uk}} \and Chao
Zhou\thanks{Department of Mathematics, Institute of Operations
Research and Analytics, and Suzhou Research Institute, National
University of Singapore , Singapore; email:
\texttt{matzc@nus.edu.sg}}}
\date{}
\maketitle

\begin{abstract}
This paper studies the properties of the optimal
portfolio-consumption strategies in a {finite horizon} robust
utility maximization framework with different borrowing and lending
rates. In particular, we allow for constraints on both investment
and consumption strategies, and model uncertainty on both drift and
volatility. With the help of explicit solutions, we quantify the
impacts of uncertain market parameters, portfolio-consumption
constraints and borrowing costs on the optimal
strategies and their time monotone properties.\\


\end{abstract}

\setcounter{equation}{0} \setcounter{section}{0}
\section{Introduction}

One of the fundamental problems in mathematical finance is the
construction of investment and consumption strategies $(\pi,c)$ that
maximize the expected utility of a risk-averse investor:
\begin{equation}\label{problem1}
\max_{(\pi,c)}E\left[\int_0^TU^c(c_s)ds+U(X_T^{\pi,c;\mu,\sigma})\right],
\end{equation}
where $U^c(\cdot)$ and $U(\cdot)$ are the utilities of intertemporal
consumption $c$ and terminal wealth $X_T^{\pi,c;\mu,\sigma}$,
respectively. The market is described by a set of parameters
$(\mu,\sigma)$--the drift and volatility of the risky assets, and
the investor's utilities are often assumed to admit some homothetic
properties (for example, power, logarithm and exponential types).
Due to the market incompleteness arising from the randomness of the
market parameters and the portfolio constraints, the resulting
optimal portfolio is described as the sum of a myopic strategy of
Merton's type and a hedging strategy. The latter is used to
partially hedge the market risk stemming from the market
incompleteness. Both the hedging strategy and the optimal
consumption can be described via the solution of a backward
stochastic differential equation (see \cite{Hu} and \cite{him}).
However, the solution is in general not explicit, and consequently,
there is limited information about the properties of the optimal
strategies.

The purpose of this article is to study the properties of the
optimal investment and consumption strategies when the investor
optimally allocates her wealth among risky and riskless assets and
her consumption. Our model takes consideration of several features
including \emph{model uncertainty, constraints on both investment
and consumption strategies, and borrowing costs}. Under both power
and logarithm utility functions, we characterize the optimal
portfolio-consumption strategies and the worst-case market
parameters using the solutions of nonlinear ODEs, and furthermore,
derive their explicit solutions in one-dimensional setting. The
explicit forms further allow us to study the impacts of uncertain
market parameters, portfolio-consumption constraints and different
borrowing and lending rates on the optimal strategies and their time
monotone properties.

In the vast majority of the literature, it is often assumed that the
investor has a perfect knowledge of the market parameters, and is
able to select her portfolio-consumption strategies without any
constraints. However, constraints such as prohibition of short
selling risky assets and the subsistence consumption are ubiquitous
in reality. On the other hand, the paradigm of expected utility
clearly has some deficiencies: it is not satisfactory in dealing
with model uncertainty as predicted by the famous Ellsberg paradox.
For the above reasons, it is desirable to take constraints on the
portfolio-consumption strategies and uncertainty about the market
parameters into account when studying the optimal strategies. We
argue that the portfolio-consumption strategies must stay in a
closed and convex set, and there are lots of probability models to
describe the market, but none of them are really precise enough.
This leads us to consider the so called \emph{robust utility
maximization} for which the investor worries about the worst-case
scenario\footnote{Note that the worst-case scenario approach implies
that the investor behaves too conservative, which is not always the
case in reality. Recently, an interesting paper \cite{Herrmann}
casts the investor having moderate risks and uncertainty aversions.
We refer to \cite{Herrmann} for a further discussion of this
approach.}, and as opposed to (\ref{problem1}), we solve the
following maxmin problem \bee\label{problem2}
\max_{(\pi,c)\in\mathcal{B}}\min_{(\mu,\sigma)\in\mathcal{A}}E\left[\int_0^TU^c(c_s)ds+U(X_T^{\pi,c;\mu,\sigma})\right],
\eee for an investor with power or logarithm type utilities on both
intertemporal consumption and terminal wealth. See
(\ref{objectfunctional1}) and (\ref{problem}) for further details.

As a first contribution, we show that the functions used to
construct the value processes for power and logarithm utilities (see
$F_P$ and $F_L$ in (\ref{definition of F})) admit saddle points
(Lemma \ref{lemma for saddle point}). The saddle points in turn
characterize locally the optimal portfolio-consumption strategies
and the worst-case parameters. Since the constraint set for the
portfolios-consumption strategies may not be compact, it is not even
clear \emph{ex ante} whether a saddle point exists or not. We tackle
the problem by an approximation procedure using a sequence of saddle
points in compact sets to construct a saddle point in the
non-compact constraint set. We further characterize the optimal
strategies using the solutions of nonlinear ODEs in Theorem
\ref{Theorem1 for power} (for power utility) and Theorem
\ref{Theorem1 for log} (for logarithm utility). We show that even
with random market parameters and portfolio{-consumption}
constraints, the optimal strategies are however deterministic in a
robust utility framework. It is due to the fact that when the
investor worries about the worst-case scenario, the optimal
strategies are given via a deterministic saddle point and the
solution of an associated nonlinear ODE. Eventually, this leads the
investor to implement myopic strategies of Merton's type to optimize
her portfolios as in a complete market. Thus, there is no need for
her to enforce the hedging strategy as opposed to the incomplete
market situation. A similar phenomenon also occurs in
\cite{Neufeld}, where the authors considered a market driven by
L\'evy processes with uncertain parameters but without consumption
and borrowing costs.

Furthermore, in one-dimensional setting we obtain the optimal
portfolio-consumption strategies and the worst-case parameters both
in closed forms. Closed-form solutions seldom exist except for the
standard Merton's
model with constant market parameters 
without portfolio-consumption constraints. We find that the explicit
solutions still exist for both power and logarithm utility functions
in the general framework incorporating model uncertainty,
constraints on both investment and consumption strategies, and
borrowing costs.

As the first example, when the uncertain market parameters stay in
an interval set, we obtain a classification of the optimal portfolio
strategies in terms of borrowing and lending rates as well as the
uncertain market parameters. We show that (1) when the investor is
optimistic about the market, meaning that her worst estimation of
the stock's return is still better than the borrowing rate, she will
implement a \emph{borrow-to-buy strategy} to borrow as much as
possible {to approach the optimal strategy without constraint}. (2)
When her worst estimation of the stock's return is between the
borrowing and lending rates, neither borrowing nor lending are
attractive, and the investor will simply put all her money in the
stock, i.e. performing \emph{a full-position strategy}. (3) When the
lending rate is between the best and worst estimations of the
stock's return, the investor will simply put all her money in the
bank account, i.e. performing \emph{a no-trading strategy}. (4) When
the investor is pessimistic about the market, meaning that her best
estimation of the stock's return is still lower than the lending
rate, she will implement a \emph{shortsale strategy} to short sell
the stock as much as possible. See Theorem \ref{theorem_portfolio}
for further details.

As the second example, when the uncertain drift and volatility are
correlated, we further show that the saddle point may become an
interior point of the uncertain parameter set. The worst-case
parameters are then given through the explicit interior saddle
point, as opposed to the bang-bang type of saddle points in the
existing literature. As a result, the optimal portfolio strategy is
also given through the interior saddle point, albeit still in
Merton's type. See Theorem \ref{theorem for ambiguity correlation}
for further details.

The explicit solutions further allows us, for the first time, to
give a systematic study of the consumption plans in various
situations. We argue that the consumption should stay above a
minimum level for subsistence purpose, and be dominated by a
reasonable upper bound for the sake of future consumption and
investment. We show that the investor's optimal consumption will
degenerate to a deterministic process {when} she worries about the
worst-case market scenario (see Theorem \ref{theorem_consumption}
for the power case and Theorem \ref{theorem_consumption for log} for
the logarithm case). By virtue of the closed form solutions, we are
able to obtain the time monotone properties of the optimal
consumption plan (see Proposition \ref{proposition} and Theorem
\ref{theorem_consumption for log}), and quantify the impacts of
different parameters (e.g. borrowing rate, uncertain market
parameters and portfolio-consumption constraints) on the optimal
consumption plan (see Propositions \ref{proposition2} and
\ref{proposition3}).

One of the striking results is that, for the power utility case, the
optimal consumption is not necessarily increasing or decreasing when
the investor lifts her upper bound for consumption. This is because
the investor needs to balance her current consumption and future
consumption and investment when she optimizes her consumption plans.
Increasing the upper bound of consumption means the investor would
consume in a larger constraint set in the future, and increase the
weight of her future utility, thus the investor might decrease her
current consumption level. On the other hand, lifting the upper
bound for consumption also means a larger constraint set from which
the investor makes her current consumption decisions, and in turn
her current consumption level might increase. This two contradicting
factors will offset their impacts by each other, and result in a
non-monotone relationship of optimal consumption with respect to the
upper bound of consumption plans.

Turning to the literature, optimal portfolio-consumption problems in
continuous time were first studied by Merton in 1970s (see
\cite{Merton} for a summary). In a sequence of papers
\cite{Karatzas}, \cite{Karatzas2} and \cite{Lehoczky}, the authors
 developed and generalized Merton's model. In particular,
\cite{Lehoczky} is one of the first arguing that the consumption
must always be above a certain subsistence level, and sometimes
neither borrowing nor shortsale are allowed for trading stocks, so
they imposed constraints on both consumption and investment.
Following this work, the optimal consumption with constraints was
further studied in \cite{Cadenillas}, \cite{Sethi}, and more
recently in \cite{Yi2}, \cite{Yi1} in a complete market setting with
constant market parameters. On the other hand, \cite{Cvitanic},
\cite{Vila}, \cite{Yan} and \cite{Zarihopoulou} among others studied
constrained investment problems for models of varying generality.

Equal borrowing and lending rates is often assumed in the
literature, and as a consequence, the wealth equation is always
linear. However, it is argued in \cite{Berman} that such an
assumption stands in contrast with reality. Subsequently,
\cite{Fleming} introduced the borrowing cost for the utility
maximization problem, and more recently in \cite{Bo}, the authors
took borrowing costs into account in an optimal credit investment
problem.

The early development of model uncertainty went back to \cite{Talay}
where the authors considered a worst-case risk management problem.
Robust utility maximization in mathematical finance started with
\cite{Bordigoni}, \cite{Hernandez} and \cite{Schied}, which mainly
dealt with drift uncertainty. The problem of volatility uncertainty
is much harder, and has been treated via various mathematical tools.
To name a few, duality method was used in \cite{Denis} where the
uncertainty is specified by a family of semimartingales laws.
$G$-expectation was employed in \cite{Forque} in a stochastic
volatility model to treat uncertain correlations. In contrast,
\cite{Matoussi} studied the robust utility maximization problem
under volatility uncertainty via second-order backward stochastic
differential equations, and \cite{Tevzadzae} considered uncertain
drift and volatility using mixed strategies and derived an explicit
solution in a non-traded asset setting. More recently, the results
have been further generalized in \cite{Neufeld} to include drift,
volatility and jump uncertainty, which are parameterized by a set of
L\'evy triples. However, consumption is not considered in the above
works. Two exceptions are \cite{Lin} and more recently
\cite{Biagini}, where the authors worked in a similar framework to
our model, but portfolio-consumption constraints are not treated in
those papers.

In summary, it seems the existing literature mainly focuses on the
investment-consumption models with only parts of the above features:
either with portfolio constraints and market uncertainty or with
consumption constraints and borrowing costs. Although many elegant
mathematical results are achieved in these papers, explicit
solutions and the properties of the optimal strategies rarely exist
except for some special cases. In particular, consumption
constraints make it difficult to obtain explicit solutions, and
almost all
of the explicit solutions with consumption constraints 
are in the framework of infinite horizon (see \cite{Cadenillas},
\cite{Lehoczky} and \cite{Sethi}).

In contrast, our paper systematically studies constrained
portfolio-consumption strategies under model uncertainty and
borrowing costs in a finite horizon, and quantifies their impacts on
the optimal strategies. We obtain explicit solutions and properties
of the optimal strategies. Although explicit solutions are derived
under one risky asset setting, our method can be applied to study
the multiple risky assets setting as in \cite{Neufeld}, and similar
results will still hold, albeit with more complicated situations.

The paper is organized as follows. Section 2 presents a robust
utility maximization model subject to borrowing costs and
portfolio-consumption constraints in a multiple risky assets
setting. Section 3 solves the associated maxmin problem via a
martingale argument, and characterizes the optimal
portfolio-consumption strategies and the worst-case market
parameters via the solutions of nonlinear ODEs. Section 4 further
obtains their closed form solutions in a single risky asset setting
with different uncertain parameter sets. Section 5 studies the
impacts of the various model parameters on the optimal strategies
and the worst-case parameters. The proof of explicit solutions is
given in the Appendix.


\section{The utility maximization model}

\subsection{Uncertain parameters and borrowing costs}

Let $d$ and $d^\prime$ be two positive integers. Let $W$ be a
standard ${d^\prime}$-dimensional Brownian motion defined on a
complete probability space $(\Omega,{\cal F},\mathbb P)$, and
$\mathbb{F}:=\{{\cal F}_t\}_{t\geq 0}$ be the augmented filtration
generated by $W$. The market consists of $d$ risky assets and a
riskless bank account. The price processes of the risky assets
$S_{i}$, $1\leq i\leq d$, solve \be\label{riskyasset}
  dS_{i,s} = \mu_{i,s} S_{i,s} ds + \Sigma_{j=1}^{d^\prime} \sigma^{ij}_s S_{i,s}  dW_{j,s}
\ee for $s\geq 0$, where $\mu:=(\mu_1,\cdots,\mu_d)^{\rm T}$ and
$\Sigma:=(\sigma^{ij})_{d\times {d^\prime}}$ represent the drift and
volatility of the risky assets, respectively.

Consider a small investor in this market. She trades both the risky
assets and riskless bank account, yet she has limited information
about the risky assets' parameters $(\mu,\Sigma)$. The uncertainty
about drift and volatility of the risky assets is parameterized by a
nonempty set with the form \bee\label{admissiblesetb1}
 {\cal B}=\left\{(\mu_s,\Sigma_s)_{s\geq 0}:(\mu,\Sigma)\;\mbox{are}\,\mathbb{F}\text{-progressively measurable},\,
 \text{and}\ (\mu_s,\Sigma_s\Sigma^{\rm T}_s)\in\mathbb{B},\ \mathbb{P}\otimes ds\text{-a.e.} \right\},
\eee where {$\mathbb{B}$ is a convex and compact subset of
$\mathbb{R}^d\times {\cal S}^d_+$, with ${\cal S}^d_+$ being the set
of $d\times d$ positive semi-definite real symmetric matrixes. We
also assume that $\mathbb{B}$ contains at least one element
$(\mu,\Sigma)$ such that $\Sigma\Sigma^{\rm T}$ is positive
definite.} The area of the set $\mathbb{B}$ indicates the amount of
uncertainty. The larger the area, the larger becomes the set of
alternative models. The investor will then become more uncertain
about the model parameters.

In terms of the bank account $B$, the standard assumption of equal
borrowing and lending rates is in contrast with empirical evidence
(see \cite{Berman}). In reality, there always exists a spread
between borrowing and lending rates. Let $R$ and $r$ be the constant
borrowing and lending rates, respectively. When $B$ is positive, the
investor lends with rate $r$. When $B$ is negative, the investor
borrows with rate $R$. {It is nature to assume that $R\geq r$.}
Consequently, the bank account $B$ follows \be \label{riskfreeasset}
 dB_s= (r B^+_s-RB^-_s)\,ds,
\ee where $x^+=\max\{0,x\},\,x^-=\max\{0,-x\}$. Note that $r
B^+_s-RB^-_s=rB_s-(R-r)B^{-}_s$, and therefore the spread $(R-r)$
represents the borrowing cost of the investor. The larger the
spread, the more borrowing cost the investor has to bear. In the
next section, we shall see the introduction of borrowing cost leads
to a nonlinear wealth equation, which is concave in the portfolio
strategies.

\subsection{Portfolio and consumption constraints}

Let $T>0$ represent the trading horizon, and suppose that the
investor has an initial wealth $x>0$. Let $\pi$ be the
\emph{proportion of her wealth} invested in {the risky assets}, $c$
be her \emph{consumption rate proportional to her wealth}, and
$X^{x;\pi,c,\mu,\Sigma}$ be the wealth process with initial value
$x$, portfolio-consumption strategies $(\pi,c)$ and parameters
$(\mu,\Sigma)$. Using (\ref{riskyasset}) and (\ref{riskfreeasset}),
it follows from the self-financing condition that \bee \nonumber
 X^{x;\pi,c,\mu,\Sigma}_s&=&x+\int_0^s\Big[\,\mu_u^{\rm T}\pi_u
 +r(1-1_d^{\rm T}\pi_u)-(R-r)(1-1_d^{\rm T}\pi_u)^--c_u\,\Big]\,{X^{x;\pi,c,\mu,\Sigma}_u}\,du
 \\[2mm]\label{wealth}
 &&+\,\int_0^s{X^{x;\pi,c,\mu,\Sigma}_u}\pi^{\rm T}_u\Sigma_u \,dW_u,\quad
 s\in[0,T].
\eee

Note that with the borrowing cost, the drift of the wealth equation
is no longer linear but concave in the portfolio strategy $\pi$ {in
the case of $R>r$}.

The investor will select her portfolio-consumption strategies from
the the following admissible set with constraints on both portfolio
and consumption: \bee
 {\cal A}&=&\left\{\,(\pi_s,c_s)_{s\geq 0}:(\pi,c)\;\mbox{are\,$\mathbb{F}$-progressively measurable,}\,
 (\pi_s,c_s)\in \mathbb{A},\, \mathbb{P}\otimes
 ds\text{-a.e.},\right.\\
 &&\ \int_0^T\left(|\pi_s|^2+c_s\right)ds<+\infty,\ \text{and}\ X^{x;\pi,c,\mu,\Sigma}\ \text{satisfies the
 condition (H)}\},
\eee where $\mathbb{A}$ is a convex and closed subset of
$\mathbb{R}^{d+1}$ satisfying that $c\geq 0$. The integrability
condition on $(\pi,c)$ is to guarantee that the wealth process is
well defined, while the condition (H) imposed on the wealth process
$X^{x;\pi,c,\mu,\Sigma}$ depends on the utility maximization problem
that we want to solve, and will be specified in (\ref{ClassD}) in
the next section.

One typical example of the constraint set is
$\mathbb{A}=\bigotimes_{i=1}^d[\,\underline{\pi}_i,\overline{\pi}_i\,]\times[\,\underline{c},\overline{c}\,]$,
where
{$\underline{\pi}_i,\overline{\pi}_i,\underline{c},\overline{c}$ are
constants satisfying} $-\infty\leq \underline{\pi}_i\leq0,\;1\leq
\overline{\pi}_i\leq+\infty,\;0\leq
\underline{c}\leq\overline{c}\leq+\infty$ for $i=1,\cdots,d$. Then,
the portfolio constraint cube
$\bigotimes_{i=1}^d\left[\,\underline{\pi}_i,
\overline{\pi}_i\,\right]$ has the following financial
interpretations: $\left(\sum_{i=1}^d\overline{\pi}_i-1\right)$
represents the maximum proportion of wealth that the investor is
allowed to borrow to invest in the risky assets;
$\left(-\sum_{i=1}^d\underline{\pi}_i\right)$ represents the largest
shortsale position that the investor is allowed to take;
$\underline{\pi}_i=0$ means prohibition of shortsale the $i$th risky
asset; $\overline{\pi}_i=1$ means prohibition of borrowing to invest
in the $i$th risky asset; and
$-\underline{\pi}_i=\overline{\pi}_i=+\infty$ means no portfolio
constrains on the $i$th risky asset. Moreover, the consumption
constraint $[\underline{c},\overline{c}]$ means that the investor
should keep a minimal consumption level $\underline{c}$ for
subsistence purpose, and at the same time, her consumption is also
controlled by an upper bound $\overline{c}$ for the sake of future
consumption and investment.

\subsection{The robust utility maximization problem}

The investor has utilities of both intertemporal consumption and
terminal wealth. Given a portfolio-consumption strategy
$(\pi,c)\in\mathcal{A}$, her expected utility is defined as
\be\label{objectfunctional1}
 {\cal J}_i(x;\pi,c,\mu,\Sigma):=\mathbb{E}\left[\,\int_0^T \lambda e^{-\rho s}U_i^c({c_sX^{x;\pi,c,\mu,\Sigma}_s})ds
 +e^{-\rho T}U_i(\,X^{x;\pi,c,\mu,\Sigma}_T)\right],\;\;i=P,L,
\ee where $P,L$ represents, respectively, the power and logarithm
utility {functions, i.e.} $U_P^c(x)=U_P(x)=\frac{1}{p}x^{p}$ with
$p\in(-\infty,0)\cup(0,1)$, and $U_L^c(x)=U_L(x)=\ln x$. Herein,
$\lambda\geq 0$ represents the weight of the intertemporal
consumption relative to the final bequest at maturity $T$, and
$\rho\geq 0$ represents the discount factor.

Since the investor is uncertain about the model parameters
$(\mu,\Sigma)$, she will seek for an optimal portfolio-consumption
strategy that is least affected by model uncertainty. In
anticipation of the worst-case scenario, she solves the following
maxmin problem: Find $(\pi^*,c^*)\in {\cal A}$ and
$(\mu^*,\Sigma^*)\in{\cal B}$ such that
\begin{equation}\label{problem}
  J_i(x):
  =\sup\limits_{(\pi,c)\in{\cal A}}\inf\limits_{(\mu,\Sigma)\in{\cal B}}{\cal J}_i(x;\pi,c,\mu,\Sigma)
  ={{\cal J}_i(x;\pi^*,c^*,\mu^*,\Sigma^*)},\;\;i=P,L,
\end{equation}
where $J_i(\cdot)$ is the value function of the maxmin problem (\ref{problem}), i.e. the maximum worst-case expected utility.

To robustify the optimal portfolio-consumption strategy, the inner
part of the above optimization problem is played by a so called
mother nature who acts maliciously to minimize the expected utility
by choosing the worst-case scenario, whereas the investor aims to
select the best strategy that is least affected by the mother
nature's choice. For this reason, the maxmin problem (\ref{problem})
is also dubbed as the \emph{robust utility maximization problem} in
the literature (see \cite{Neufeld} for example).

To solve the value function of the robust utility maximization
problem (\ref{problem}) and its corresponding worst-case parameters
and optimal portfolio-consumption strategies, we look for a saddle
point strategy $\left\{(\pi^*,c^*),(\mu^*,\Sigma^*)\right\}$ of the
expected utility ${\cal J}_i(x;\pi,c,\mu,\Sigma)$ such that
\be\label{saddle_point} {\cal J}_i(x;\pi,c,\mu^*,\Sigma^*)\leq {\cal
J}_i(x;\pi^*,c^*,\mu^*,\Sigma^*)\leq {\cal
J}_i(x;\pi^*,c^*,\mu,\Sigma) \ee for any admissible
$(\pi,c)\in\mathcal{A}$ and $(\mu,\Sigma)\in\mathcal{B}$. Then, it
follows that
$$\sup\limits_{(\pi,c)\in{\cal
A}}\inf\limits_{(\mu,\Sigma)\in{\cal B}}{\cal
J}_i(x;\pi,c,\mu,\Sigma)={\cal
J}_i(x;\pi^*,c^*,\mu^*,\Sigma^*)=\inf\limits_{(\mu,\Sigma)\in{\cal
B}}\sup\limits_{(\pi,c)\in{\cal A}}{\cal J}_i(x;\pi,c,\mu,\Sigma),$$
and consequently, ${J}_i(x)={\cal J}_i(x;\pi^*,c^*,\mu^*,\Sigma^*)$
is the value function of the maxmin problem $(\ref{problem})$, with
$(\mu^*,\Sigma^*)$ and $(\pi^*,c^*)$ as the worst-case parameters
and the optimal portfolio-consumption strategies, respectively.

To close this section, we further specify the condition (H) in the
admissible set $\mathcal{A}$ associated with the maxmin problem
(\ref{problem}):
\begin{align}\label{ClassD}
\text{ Condition (H)}:=&\left\{E\left[\int_0^T
U^c_i(c_sX_s^{x;\pi,c,\mu,\Sigma})ds\right]<+\infty;\ \text{and the
family}\ U_i\left(X^{x;\pi,c,\mu,\Sigma}_{\tau}\right),\
\right.\notag\\
&\left. \text{for}\ \tau\in[0,T]\ \text{as an}\
\mathbb{F}\text{-stopping time, is uniformly integrable}\right\}.
\end{align}
The integrability condition imposed on
$U_i\left(X^{x;\pi,c,\mu,\Sigma}\right) $ is to include unbounded
portfolio and consumption strategies. This condition is also called
Class (D) condition and appears in \cite{Hu}, where the authors
solve a similar portfolio-consumption problem, but without model
uncertainty, borrowing costs and consumption constraints.

\section{Nonlinear ODE characterization of the value functions}

In this section, we apply a martingale argument, firstly introduced
in \cite{Hu} and \cite{him}, to construct a saddle point strategy
$\left\{(\mu^*,\Sigma^*),(\pi^*,c^*)\right\}$ for the expected
utility ${\cal J}_i(x;\pi,c,\mu,\Sigma)$. This will in turn solve the
original maxmin problem (\ref{problem}).

To this end, we aim to construct an $\mathbb{F}$-adapted process
$J_{i,t}^{x;\pi,c,\mu,\Sigma}$, $t\in[0,T]$, satisfying the following
three conditions: For any $(\pi,c)\in\mathcal{A}$ and
$(\mu,\Sigma)\in\mathcal{B}$,

(C1) at the maturity $T$,

$$J_{i,T}^{x;\pi,c,\mu,\Sigma}=\int_0^T \lambda e^{-\rho s}U^c_i\Big(\,c_s X^{x;\pi,c,\mu,\Sigma}_s\,\Big)ds
 +e^{-\rho T}U_i\Big(\,X^{x;\pi,c,\mu,\Sigma}_T\,\Big);$$

(C2) at the initial time $0$, $J_{i,0}^{x;\pi,c,\mu,\Sigma}=J_{i,0}^x$,
which is a constant and is independent of $(\pi,c)$ and
$(\mu,\Sigma)$;

(C3) there exist $(\pi^*,c^*)\in\mathcal{A}$ and
$(\mu^*,\Sigma^*)\in\mathcal{B}$ such that the process
$J_i^{x;\pi^*,c^*,\mu^*,\Sigma^*}$ is a martingale,
$J_i^{x;\pi,c,\mu^*,\Sigma^*}$ is a supermartingale, and
$J_i^{x;\pi^*,c^*,\mu,\Sigma}$ is a submartingale.

{Following the above conditions (C1-C3)}, we then have \bee {\cal
J}_i(x;\pi,c,\mu^*,\Sigma^*)&=& E[J_{i,T}^{x;\pi,c,\mu^*,\Sigma^*}]\leq
J_{i,0}^{x;\pi,c,\mu^*,\Sigma^*}=J_{i,0}^x;\\
{\cal J}_i(x;\pi^*,c^*,\mu^*,\Sigma^*)&=&
E[J_{i,T}^{x;\pi^*,c^*,\mu^*,\Sigma^*}]=J_{i,0}^{x;\pi^*,c^*,\mu^*,\Sigma^*}=J_{i,0}^x;\\
{\cal J}_i(x;\pi^*,c^*,\mu,\Sigma)&=&E[J_{i,T}^{x;\pi^*,c^*,\mu,\Sigma}]\geq
J_{i,0}^{x;\pi^*,c^*,\mu,\Sigma}=J_{i,0}^x. \eee Thus, {the inequalities in
(\ref{saddle_point}) hold, i.e.,
$\left\{(\pi^*,c^*),(\mu^*,\Sigma^*)\right\}$ is a saddle point
strategy of the expected utility ${\cal J}_i(x;\pi,c,\mu,\Sigma)$, and
the value function of the maxmin problem (\ref{problem}) is given by
$J_i(x)=J_{i,0}^x$.}

Next, we construct the process ${\cal J}_i^{x;\pi,c,\mu,\Sigma}$. We
start with the following lemma, which reduces the original maxmin
problem (\ref{problem}), which is an infinite dimensional
optimization problem, to a finite dimensional one. {To facilitate
our discussions below, we introduce two functions
$F_i(\cdot;\cdot,\cdot;\cdot,\cdot),i=P,L$, which characterize the
optimal portfolio-consumption and the worst-case parameters
locally,} \be\nonumber
 &&F_i(x_q;x_\pi,x_c;x_\mu,x_\Sigma)
 \\[2mm]\label{definition of F}
 &:=&\left\{
 \begin{array}{ll}
 {p-1\over 2}\,x_\pi^{\rm T}x_\Sigma x_\pi
 +\Big[\,x_\mu^{\rm T}x_\pi+r(1-1_d^{\rm T}x_\pi)^+-R(1-1_d^{\rm T}x_\pi)^-\,\Big]+{{\lambda\over
 p}e^{-x_q}x_c^{p}-x_c},&i=P;
 \vspace{2mm}\\
 -{1\over 2}\,x_\pi^{\rm T}x_\Sigma x_\pi
 +\Big[\,x_\mu^{\rm T}x_\pi+r(1-1_d^{\rm T}x_\pi)^+-R(1-1_d^{\rm T}x_\pi)^-\,\Big]+\lambda e^{-x_q}\ln x_c- x_c,&i=L;
 \end{array}
 \right.
\ee for $x_q\in \mathbb{R}$, $(x_\pi,x_c)\in \mathbb{A}$ and
$(x_\mu,x_\Sigma)\in \mathbb{B}$. Recall that $\mathbb{A}$ is convex
and closed, and $\mathbb{B}$ is convex and compact.

\begin{lemma}\label{lemma for saddle point}
For $i=P,L$, the function $F_i(x_q;\cdot,\cdot;\cdot,\cdot)$ admits
the following properties.

(i) The function $F_i(x_q;\cdot,\cdot;\cdot,\cdot)$ admits at least
one saddle point
$(\widetilde{x}^*_\pi(x_q),\widetilde{x}^*_c(x_q);\widetilde{x}^*_\mu(x_q),\widetilde{x}^*_\Sigma(x_q))$,
i.e. for any $x_q\in \mathbb{R}$, $(x_\pi, x_c)\in \mathbb{A}$ and
$(x_\mu,x_\Sigma)\in \mathbb{B}$, \be\nonumber
 F_i(x_q;\widetilde{x}^*_\pi(x_q),\widetilde{x}^*_c(x_q);x_\mu,x_\Sigma)
 &\geq& F_i(x_q;\widetilde{x}^*_\pi(x_q),\widetilde{x}^*_c(x_q);\widetilde{x}^*_\mu(x_q),\widetilde{x}^*_\Sigma(x_q))
 \\[2mm]\label{saddle point of F}
 &\geq&
 F_i(x_q;x_\pi,x_c;\widetilde{x}^*_\mu(x_q),\widetilde{x}^*_\Sigma(x_q)).
\ee

(ii) For $x_q\in\mathbb{R}$, let \be\label{defintion of G}
G_i(x_q):=F_i(x_q,\widetilde{x}_\pi^*(x_q),\widetilde{x}_c^*(x_q);\widetilde{x}_\mu^*(x_q),\widetilde{x}_\Sigma^*(x_q)).
\ee Then,
$G_i(x_q),\widetilde{x}_\pi^*(x_q),\widetilde{x}_c^*(x_q),\widetilde{x}_\mu^*(x_q)$
and $\widetilde{x}_\Sigma^*(x_q)$ are locally bounded in
$x_q\in\mathbb{R}$.

(iii) If $p<0$ or $i=L$, then $(\widetilde{x}_c^*(x_q))^{-1}$ is
also locally bounded in $x_q\in\mathbb{R}$.
\end{lemma}

\begin{proof}
\emph{Step 1.} We first prove the assertion (i) when the set
$\mathbb{A}$ is compact. Indeed, for fixed $x_q\in \mathbb{R}$, it
is clear that the function $F_i(x_q;\cdot,\cdot;\cdot,\cdot)$ is
concave with respect to $(x_\pi, x_c)$, and convex (accurately
linear) with respect to $(x_\mu,x_\Sigma)$. Since $\mathbb{A}$ and
$\mathbb{B}$ are convex and compact, we may apply the minmax theorem
(see Theorem B on pp. 131 in \cite{Roberts} or section 3 in
\cite{Neufeld}), and deduce that there exists a saddle point
$(\widetilde{x}^*_\pi,\widetilde{x}^*_c;\widetilde{x}^*_\mu,\widetilde{x}^*_\Sigma)$
such that~\eqref{saddle point of F} holds. Moreover, the compactness
of $\mathbb{A}$ and $\mathbb{B}$ implies that
$\widetilde{x}^*_\pi,\widetilde{x}^*_c,\widetilde{x}^*_\mu,\widetilde{x}^*_\Sigma$
are bounded.\smallskip

\emph{Step 2.} If the set $\mathbb{A}$ is not compact, for any
positive integer $n$, let
$\mathbb{A}_n:=\mathbb{A}\cap\{(x_\pi,x_c):|(x_\pi,x_c)|\leq n\}$.
It is clear that we can choose a large enough positive integer $N$
such that $\mathbb{A}_n$ is non-empty for any $n\geq N$ and, without
loss of generality, we may suppose that $n\geq N$ below. Thanks to
Step 1, we know that the function $F_i(x_q;\cdot,\cdot;\cdot,\cdot)$
has at least one saddle point
$(\widetilde{x}^n_\pi,\widetilde{x}^n_c;\widetilde{x}^n_\mu,\widetilde{x}^n_\Sigma)$
in $\mathbb{A}_n\times\mathbb{B}$, and we denote
$F_i(x_q;\widetilde{x}^n_\pi,\widetilde{x}^n_c;x^n_\mu,x^n_\Sigma)$
by $F^n_i$.

Next, we prove that $F^n_i$ is nondecreasing with respect to $n$ and
has a uniformly lower bound for any $n\geq N$. To this end, note
that \be\nonumber
 F^n_i&=&\inf\limits_{(x_\mu,x_\Sigma)\in \mathbb{B}}\sup\limits_{(x_\pi,x_c)\in \mathbb{A}_n} F_i(x_q;x_\pi,x_c;x_\mu,x_\Sigma)
 =\sup\limits_{(x_\pi,x_c)\in \mathbb{A}_n}\inf\limits_{(x_\mu,x_\Sigma)\in \mathbb{B}} F_i(x_q;x_\pi,x_c;x_\mu,x_\Sigma)
 \\[2mm]\label{saddle point of Fn}
 &=&\sup\limits_{(x_\pi,x_c)\in \mathbb{A}_n} F_i(x_q;x_\pi,x_c;\widetilde{x}^n_\mu,\widetilde{x}^n_\Sigma)
 =\inf\limits_{(x_\mu,x_\Sigma)\in \mathbb{B}} F_i(x_q;\widetilde{x}^n_\pi,\widetilde{x}^n_c;x_\mu,x_\Sigma).
\ee From the first equality in~\eqref{saddle point of Fn}, we deduce
that $F^n_i$ is nondecreasing with respect to $n$. Furthermore, the
second equality in \eqref{saddle point of Fn} implies that, for any
$n\geq N$ and  $(x^0_\pi,x^0_c)\in\mathbb{A}_N$,
\be\label{lowerbound}
  F^n_i\geq F^N_i\geq
  \inf\limits_{(x_\mu,x_\Sigma)\in \mathbb{B}}F_i(x_q;x^0_\pi,x^0_c;x_\mu,x_\Sigma)
  >-\infty,
\ee where we have used the fact that $\mathbb{B}$ is compact in the
last inequality. Until now, we have proved that $F^n_i$ is
nondecreasing with respect to $n$ and has a uniformly lower bound
for any $n\geq N$.

\emph{Step 3.} We prove that there exists a large enough positive
integer $M$ such that $(\widetilde{x}^n_\pi,\widetilde{x}^n_c)\in
\mathbb{A}_M$ for any $n\geq M$. Indeed, we may choose a positive
constant $\epsilon$ and a positive-definite matrix $x^0_\Sigma$ such
that $(x^0_\mu,x^0_\Sigma)\in\mathbb{B}$ and $x_\pi^{\rm
T}x_\Sigma^0 x_\pi\geq \epsilon|x_\pi|^2$ for any $x_\pi\in
\mathbb{R}^d$. Hence, as $x_c\rightarrow0^+$ when $p<0$ or $i=L$, or
$|(x_\pi,x_c)|\rightarrow+\infty$, the compactness of $\mathbb{B}$
implies that \be\label{coercivity}
 \left\{
 \begin{array}{ll}
 F_P(x_q;x_\pi,x_c;x^0_\mu,x^0_\Sigma)
 \leq \left({p-1\over 2}\,\epsilon|x_\pi|^2+C|x_\pi|\right)+\left({{\lambda\over p}e^{-x_q}x_c^{p}-x_c}\right)\rightarrow-\infty,
 \vspace{2mm}\\
 F_L(x_q;x_\pi,x_c;x^0_\mu,x^0_\Sigma)
 \leq \left(-{\epsilon\over2}|x_\pi|^2+C|x_\pi|\right)+\left(\lambda e^{-x_q}\ln x_c- x_c\right)\rightarrow-\infty,
 \end{array}
 \right.
\ee for any $(x_\pi,x_c)\in\mathbb{A}$, where $C$ is a constant
independent of $x_q,x_\pi,x_c$, $x_\mu$ and $x_\Sigma$, In turn,
there exists a large enough positive integer $M\geq N$ such that for
any $(x_\pi,x_c)\in \mathbb{A}\backslash\mathbb{A}_M$, or for any
$(x_\pi,x_c)\in \mathbb{A}$ with $x_c<1/M$ when $p<0$ or $i=L$,
$$
  \inf\limits_{(x_\mu,x_\Sigma)\in \mathbb{B}} F_i(x_q;x_\pi,x_c;x_\mu,x_\Sigma)
  \leq F_i(x_q;x_\pi,x_c;x^0_\mu,x^0_\Sigma)<F^N_i\leq F^M_i\leq F^n_i,\quad n\geq M.
$$
For the last two inequalities, we have used the fact that $F^n_i$ is
nondecreasing with respect to $n$ (see Step 2). Thus, the last
equality in~\eqref{saddle point of Fn} implies that
$(\widetilde{x}^n_\pi,\widetilde{x}^n_c)\in \mathbb{A}_M$ for any
$n\geq M$ and, moreover, $\widetilde{x}^n_c\geq 1/M$ when $p<0$ or
$i=L$.

\emph{Step 4.}  We prove that the function
$F_i(x_q;\cdot,\cdot;\cdot,\cdot)$ has at least one saddle point
$(\widetilde{x}^*_\pi,\widetilde{x}^*_c;\widetilde{x}^*_\mu,\widetilde{x}^*_\Sigma)$
in $\mathbb{A}\times\mathbb{B}$. Indeed, according to Step 3,
$F_i(x_q;\cdot,\cdot;\cdot,\cdot)$ has at least one saddle point
$(\widetilde{x}^n_\pi,\widetilde{x}^n_c;\widetilde{x}^n_\mu,\widetilde{x}^n_\Sigma)$
in $\mathbb{A}_n\times\mathbb{B}$, and all of them belong to a
compact set $\mathbb{A}_M\times\mathbb{B}$ for any $n\geq M$. Hence,
there exists a subsequence (still denoted by itself) such that
$(\widetilde{x}^n_\pi,\widetilde{x}^n_c;\widetilde{x}^n_\mu,\widetilde{x}^n_\Sigma)
\rightarrow
(\widetilde{x}^*_\pi,\widetilde{x}^*_c;\widetilde{x}^*_\mu,\widetilde{x}^*_\Sigma)\in
\mathbb{A}_M\times\mathbb{B} \subseteq\mathbb{A}\times\mathbb{B}$.
Next, we prove that
$(\widetilde{x}^*_\pi,\widetilde{x}^*_c;\widetilde{x}^*_\mu,\widetilde{x}^*_\Sigma)$
is a saddle point of $F_i(x_q;\cdot,\cdot;\cdot,\cdot)$ in
$\mathbb{A}\times\mathbb{B}$.

It clear that \begin{equation}\label{inequ}
 F_i(x_q;\widetilde{x}^n_\pi,\widetilde{x}^n_c;x_\mu,x_\Sigma)
 \geq F_i(x_q;\widetilde{x}^n_\pi,\widetilde{x}^n_c;\widetilde{x}^n_\mu,\widetilde{x}^n_\Sigma)
 \geq F_i(x_q;x_\pi,x_c;\widetilde{x}^n_\mu,\widetilde{x}^n_\Sigma),
\end{equation} for any $(x_\pi, x_c)\in \mathbb{A}_n$ and
$(x_\mu,x_\Sigma)\in \mathbb{B}$. Sending $n\rightarrow+\infty$ in
the first inequality in (\ref{inequ}), we deduce that for any
$(x_\mu,x_\Sigma)\in \mathbb{B}$,
$$
  F_i(x_q;\widetilde{x}^*_\pi,\widetilde{x}^*_c;x_\mu,x_\Sigma)
 \geq F_i(x_q;\widetilde{x}^*_\pi,\widetilde{x}^*_c;\widetilde{x}^*_\mu,\widetilde{x}^*_\Sigma).
$$
On the other hand, for any $(x_\pi, x_c)\in \mathbb{A}$, we can
choose a large enough positive integer $\widetilde{N}$ such that
$(x_\pi, x_c)\in \mathbb{A}_n$ for any $n\geq \widetilde{N}$. Then,
sending $n\rightarrow+\infty$ in the second inequality  in
(\ref{inequ}), we deduce that
$$
   F_i(x_q;\widetilde{x}^*_\pi,\widetilde{x}^*_c;\widetilde{x}^*_\mu,\widetilde{x}^*_\Sigma)
 \geq F_i(x_q;x_\pi,x_c;x^*_\mu,x^*_\Sigma).
$$
Therefore,
$(\widetilde{x}^*_\pi,\widetilde{x}^*_c;\widetilde{x}^*_\mu,\widetilde{x}^*_\Sigma)$
is a saddle point of $F_i(x_q;\cdot,\cdot;\cdot,\cdot)$ in
$\mathbb{A}\times\mathbb{B}$.

\emph{Step 5.} We prove that assertions (ii) and (iii) hold. Indeed,
from the proof in Step 4, we know that all saddle points
$(\widetilde{x}^n_\pi,\widetilde{x}^n_c;\widetilde{x}^n_\mu,
\widetilde{x}^n_\Sigma)$ belong to a compact set
$\mathbb{A}_M\times\mathbb{B}$ for any $n\geq M$ and $x_q\in
\mathbb{R}$. Furthermore, it follows from \eqref{lowerbound}
and~\eqref{coercivity} in Step 3 that, there exists a neighborhood
of $x_q$, say $x_q\in(a,b)$, such that the subscript $M$ in
$\mathbb{A}_M$ is independent of $x_q$ (but may depend on $a$ and
$b$). Thus, for $x_q\in(a,b)$,
$(\widetilde{x}_\pi^*(x_q),\widetilde{x}_c^*(x_q),\widetilde{x}_\mu^*(x_q),\widetilde{x}_\Sigma^*(x_q))\in\mathbb{A}_M\times\mathbb{B}$
which means the functions
$\widetilde{x}_\pi^*(x_q),\widetilde{x}_c^*(x_q),\widetilde{x}_\mu^*(x_q)$
and $\widetilde{x}_\Sigma^*(x_q)$ are locally bounded and, moreover,
\eqref{definition of F} and \eqref{defintion of G} imply that
$G_i(x_q)$ is also locally bounded in $x_q\in\mathbb{R}$.

From Step 3, we know that for any $n\geq M$, $\widetilde{x}^n_c\geq
1/M$ when $p<0$ or $i=L$. Since the saddle point
$(\widetilde{x}^*_\pi,\widetilde{x}^*_c;\widetilde{x}^*_\mu,\widetilde{x}^*_\Sigma)$
is the limit of
$(\widetilde{x}^n_\pi,\widetilde{x}^n_c;\widetilde{x}^n_\mu,\widetilde{x}^n_\Sigma)$,
we deduce that $\widetilde{x}^*_c\geq 1/M$ in the case of $p<0$ or
$i=L$, which means that $(\widetilde{x}^*_c(x_q))^{-1}$ is locally
bounded in $x_q\in\mathbb{R}$.
\end{proof}

We are now ready to state our first main result, which is about
nonlinear ODE characterization of the value functions $J_i(\cdot)$
for $i=P,L$. Since the conclusions for power and logarithm utility
functions are different, we present their results separately.

\begin{theorem}\label{Theorem1 for power}  Suppose that $q_P(\cdot)$ solves the following
nonlinear ODE \be\label{intergralequation}
 q_P(t)=\int_t^T\left[\,pG_P(q_P(s))-\rho\right]ds,\;\;t\in[0,T],
\ee where the function $G_P(\cdot)$ is given in Lemma \ref{lemma for
saddle point}.

Then, for the power utility case, the process
\be\label{valueprocess} J_{P,t}^{x;\pi,c,\mu,\Sigma}:= {1\over
p}\int_0^t \lambda e^{-\rho
s}\Big(\,c_sX^{x;\pi,c,\mu,\Sigma}_s\,\Big)^pds
 +{1\over p}e^{q_P(t)-\rho t}\Big(\,X^{x;\pi,c,\mu,\Sigma}_t\,\Big)^p,
\ee together with
$(\pi^*_t,c^*_t)=(\widetilde{x}_{\pi}^*(q_P(t)),\widetilde{x}_c^*(q_P(t)))$
and
$(\mu_t^*,\Sigma_t^*(\Sigma_t^*)^T)=(\widetilde{x}_{\mu}^*(q_P(t)),\widetilde{x}_{\Sigma}^*(q_P(t)))$,
$t\in[0,T]$, satisfy the conditions (C1-C3), where
$(\widetilde{x}^*_\pi(x_q),\widetilde{x}^*_c(x_q);\widetilde{x}^*_\mu(x_q),\widetilde{x}^*_\Sigma(x_q))$
is a saddle point given in Lemma \ref{lemma for saddle point}. In
particular, the value function of the maximin problem
(\ref{problem}) is given by \bee J_{P}(x)=J_{P,0}^{x}={x^p\over
p}e^{q_P(0)}. \eee
\end{theorem}

\begin{proof}

$J_{P}^{x;\pi,c,\mu,\Sigma}$  in (\ref{valueprocess}) obviously
satisfies the conditions (C1) and (C2), so it suffices to verify the
martingale property (C3).

To this end, for any $(\pi,c)\in\mathcal{A}$ and
$(\mu,\Sigma)\in\mathcal{B}$, an application of It\^o's formula
implies
\be\nonumber
 d \left(X_s^{x;\pi,c,\mu,\Sigma}\right)^p&=&
\left(X_s^{x;\pi,c,\mu,\Sigma}\right)^{p}\,\left\{\left[\,p\,F_P\left(q_P(s);\pi_s,c_s;\mu_s,\Sigma_s\Sigma_s^{\rm T}\right)
 -\lambda e^{-q_P(s)}c_s^p\,\right]\,ds\,\right.
 \\[2mm]\label{estimation}
 &&\left.+\,{p\pi_s^{\rm T}\Sigma_sdW_s}\right\},
 \ee
 and in turn,
\bee
J_{P,t}^{x;\pi,c,\mu,\Sigma}&\!\!\!=\!\!\!&
{J_{P,0}^{x;\pi,c,\mu,\Sigma}}+\int_0^te^{q_P(s)-\rho
s}(X_s^{x;\pi,c,\mu,\Sigma})^{p}\left[F_P\left(q_P(s);\pi_s,c_s;\mu_s,\Sigma_s\Sigma_s^{\rm T}\right)+{q_P^{\prime}(s)-\rho\over p}\,\right]\,ds\\
&\!\!\!\!\!\!&+\int_0^te^{q_P(s)-\rho
s}(X_s^{x;\pi,c,\mu,\Sigma})^{p}\pi_s^{\rm T}\Sigma_sdW_s.
\eee

{Since $q_P(\cdot)$ is a continuous and deterministic function, we
know that $q_P$ is bounded in the interval $[0,T]$. Together with
Lemma \ref{lemma for saddle point},} {we deduce that
$G_P(q_P(\cdot))$ and $\pi^*,c^*,\mu,\Sigma^*$ are all bounded, and
$(c^*)^{-1}$ is also bounded when $p<0$}. It follows that the
stochastic exponential
$\mathcal{E}\left(p\int_0^{\cdot}(\pi^*_s)^{\rm
T}\Sigma_s^*dW_s\right)$ is a uniformly integrable martingale.
Moreover, from (\ref{estimation}), we deduce that
$$
 \left(X_t^{x,\pi^*,c^*,\mu^*,\Sigma^*}\right)^p
 =x^p\mathcal{E}_t\left(p\int_0^{\cdot}(\pi^*_s)^{\rm T}\Sigma^*_sdW_s\right)
 \exp\left(\int_0^t\left[pG_P(q_P(s))
 -\lambda e^{-q_P(s)}(c^*_s)^p\right]ds\right)
$$
for $t\in[0,T]$. Moreover, there exists a constant $C>0$ such that
$$
 E\left[\int_0^T\left(c_t^*X_t^{x,\pi^*,c^*,\mu^*,\Sigma^*}\right)^pdt\right]\leq
 CE\left[\int_0^T\mathcal{E}_t\left(p\int_0^{\cdot}(\pi^*_s)^{\rm T}\Sigma^*_sdW_s\right)dt\right]=CT.
$$
Thus, $X^{x,\pi^*,c^*,\mu^*,\Sigma^*}$ satisfies the condition (H),
and $(\pi^*,c^*)\in\mathcal{A}$ and
$(\mu^*,\Sigma^*)\in\mathcal{B}$. Together with ODE
(\ref{intergralequation}) for $q_P(\cdot)$, we deduce that
$$E\left[J_{P,s}^{x;\pi^*,c^*,\mu^*,\Sigma^*}|\mathcal{F}_t\right]=J_{P,t}^{x;\pi^*,c^*,\mu^*,\Sigma^*}$$
for any $0\leq t\leq s\leq T$.

With
$(\mu_s^*,\Sigma_s^*(\Sigma_s^*)^T)=(\widetilde{x}_{\mu}^*(q_P(s)),\widetilde{x}_{\Sigma}^*(q_P(s)))$,
the second inequality in the saddle point condition (\ref{saddle
point of F}) implies \bee
F_P(q_P(s);\pi_s,c_s;\mu_s^*,\Sigma_s^*(\Sigma_s^*)^{\rm
T})+{q_P^{\prime}(s)-\rho\over p} \leq
G_P(q_P(s))+{q_P^{\prime}(s)-\rho\over p}=0 \eee for any
$(\pi,c)\in\mathcal{A}$. Thus $J_{P}^{x;\pi,c,\mu^*,\Sigma^*}$ is a
local supermartingale. Take an increasing sequence of
$\mathbb{F}$-stopping times $\tau_n\uparrow T$ such that {for any
$0\leq t\leq s\leq T$,}
$$E\left[J_{P,s\wedge\tau_n}^{x;\pi,c,\mu^*,\Sigma^*}|\mathcal{F}_t\right]\leq J_{P,t\wedge\tau_n}^{x;\pi,c,\mu^*,\Sigma^*},$$
i.e. \be\label{condition_inequality}
E\left[J_{P,s\wedge\tau_n}^{x;\pi,c,\mu^*,\Sigma^*}1_{A}\right]\leq
E\left[J_{P,t\wedge\tau_n}^{x;\pi,c,\mu^*,\Sigma^*}1_A\right]\ee for
{any} $A\in\mathcal{F}_t$. By the condition (H) on
$X^{x,\pi,c,\mu^*,\Sigma^*}$, we may let $\tau_n\uparrow T$ in
(\ref{condition_inequality}), which then implies that
$E[J_{P,s}^{x;\pi,c,\mu^*,\Sigma^*}1_{A}]\leq
E[J_{P,t}^{x;\pi,c,\mu^*,\Sigma^*}1_A]$, i.e.
$J_{P}^{x;\pi,c,\mu^*,\Sigma^*}$ is a supermartingale.

Finally, with
$(\pi^*_s,c^*_s)=(\widetilde{x}_{\pi}^*(q_P(s)),\widetilde{x}_c^*(q_P(s)))$,
the first inequality in the saddle point condition (\ref{saddle
point of F}) implies \bee
F_P(q_P(s);\pi^*_s,c^*_s;\mu_s,\Sigma_s\Sigma_s^{\rm
T})+{q_P^{\prime}(s)-\rho\over p} \geq
G_P(q_P(s))+{q_P^{\prime}(s)-\rho\over p}=0 \eee for any
$(\mu,\Sigma)\in\mathcal{B}$, so $J_{P}^{x;\pi^*,c^*,\mu,\Sigma}$ is
a local submartingale. Following along similar argument as above, we
obtain that $J_{P}^{x;\pi^*,c^*,\mu,\Sigma}$ is a submartingale.
\end{proof}

\begin{theorem}\label{Theorem1 for log} Suppose that $q_L(\cdot)$ and $Q_L(\cdot)$ solve the following
ODEs
 \be\label{intergralequation1 for log}
 q_L(t)=\int_t^T\left(\,\lambda e^{-q_L(s)}-\rho\right)ds,\quad
 Q_L(t)=\int_t^T e^{q_L(s)-\rho s}G_L(q_L(s))ds,\;\;t\in[0,T],
\ee where the function {$G_L(\cdot)$} is given in Lemma \ref{lemma
for saddle point}.

Then, for the logarithm utility case, the process
\be\label{valueprocess for log} J_{L,t}^{x;\pi,c,\mu,\Sigma}:=
\int_0^t \lambda e^{-\rho
s}\ln\Big(\,c_sX^{x;\pi,c,\mu,\Sigma}_s\,\Big)ds
 +e^{q_L(t)-\rho t}\ln\Big(\,X^{x;\pi,c,\mu,\Sigma}_t\,\Big)+Q_L(t),
\ee together with
$(\pi^*_t,c^*_t)=(\widetilde{x}_{\pi}^*(q_L(t)),\widetilde{x}_c^*(q_L(t)))$
and
$(\mu_t^*,\Sigma_t^*(\Sigma_t^*)^T)=(\widetilde{x}_{\mu}^*(q_L(t)),\widetilde{x}_{\Sigma}^*(q_L(t)))$,
$t\in[0,T]$, satisfy the conditions (C1-C3), where
$(\widetilde{x}^*_\pi(x_q),\widetilde{x}^*_c(x_q);\widetilde{x}^*_\mu(x_q),\widetilde{x}^*_\Sigma(x_q))$
is a saddle point given  in Lemma \ref{lemma for saddle point}. In
particular, the value function of the maximin problem
(\ref{problem}) is given by \bee J_L(x)=J_{L,0}^{x}=e^{q_L(0)}\ln
x+Q_L(0). \eee
\end{theorem}

\begin{proof}

$J_L^{x;\pi,c,\mu,\Sigma}$  in (\ref{valueprocess for log}) obviously
satisfies the conditions (C1) and (C2), so it suffices to verify the
martingale property (C3).

To this end, for any $(\pi,c)\in\mathcal{A}$ and
$(\mu,\Sigma)\in\mathcal{B}$, an application of It\^o's formula
implies
\be\label{estimation for log}
 d \ln\left(X_s^{x;\pi,c,\mu,\Sigma}\right)
 =\left[F_L\left(q_L(s);\pi_s,c_s;\mu_s,\Sigma_s\Sigma_s^{\rm T}\right)
 -\lambda e^{-q_L(s)}\ln c_s\right]ds+\pi_s^{\rm T}\Sigma_sdW_s,
 \ee
 and in turn,
\bee J_{L,t}^{x;\pi,c,\mu,\Sigma}&\!\!\!=\!\!\!&
{J_{L,0}^{x;\pi,c,\mu,\Sigma}}+\int_0^t{e^{q_L(s)-\rho
s}\left\{\,\left[\,F_L\left(q_L(s);\pi_s,c_s;\mu_s,\Sigma_s\Sigma_s^{\rm
T}\right)+e^{-q_L(s)+\rho s}Q_L^\prime(s)\,\right]\right.}
\\
&\!\!\!\!\!\!&\left.+\,\left[\,q_L^{\prime}(s)-\rho+\lambda e^{-q_L(s)}\,\right]\,
\ln X_s^{x;\pi,c,\mu,\Sigma}\,\right\}ds+\int_0^te^{q_L(s)-\rho
s}\pi_s^{\rm T}\Sigma_sdW_s
\\[2MM]
&\!\!\!=\!\!\!&
{J_{L,0}^{x;\pi,c,\mu,\Sigma}}+\int_0^t\left[\,e^{q_L(s)-\rho s}F_L\left(q_L(s);\pi_s,c_s;\mu_s,\Sigma_s\Sigma_s^{\rm T}\right)+Q_L^\prime(s)\,\right]ds
\\
&\!\!\!\!\!\!&+\int_0^te^{q_L(s)-\rho s}\pi_s^{\rm T}\Sigma_sdW_s.
\eee {Since $q_L(\cdot),Q_L(\cdot)$ are continuous and deterministic
functions, we know that $q_L(\cdot),Q_L(\cdot)$ are bounded in the
interval $[0,T]$. Together with Lemma \ref{lemma for saddle point},}
{we deduce that $G_L(q_L(\cdot))$ and
$\pi^*,c^*,\mu,\Sigma^*,(c^*)^{-1}$ are all bounded}. It follows
that the stochastic integral $\int_0^\cdot(\pi_s^*)^{\rm
T}\Sigma_s^*dW_s$ is a uniformly integrable martingale. Moreover,
from (\ref{estimation for log}), we deduce that
$$
 \ln \left(X_t^{x,\pi^*,c^*,\mu^*,\Sigma^*}\right)
 =\ln x+\int_0^t\left[G_L(q_L(s))-\lambda e^{-q_L(s)}\ln c_s\right]ds+\int_0^t\pi_s^{\rm
 T}{\Sigma_sdW_s}
$$
for $t\in[0,T]$. Since $|\ln c^*|\leq c^*+(c^*)^{-1}$,
{and there} exists a constant $C>0$ such that
$$
 E\left[\int_0^T\left|\ln (c_s^*X_s^{x,\pi^*,c^*,\mu^*,\Sigma^*})\right|ds\right]\leq
 E\left[\int_0^T\left|\ln c_s^*\right|ds\right]+E\left[\int_0^T\left|\ln X_s^{x,\pi^*,c^*,\mu^*,\Sigma^*}\right|ds\right]\leq
 {CT},
$$
we deduce that $X^{x,\pi,c,\mu^*,\Sigma^*}$ satisfies the condition
(H), and $(\pi^*,c^*)\in\mathcal{A}$ and
$(\mu^*,\Sigma^*)\in\mathcal{B}$. Together with the two ODEs
(\ref{intergralequation1 for log}) for $q_L(\cdot)$ and
$Q_L(\cdot)$, we obtain
$$E\left[J_{L,s}^{x;\pi^*,c^*,\mu^*,\Sigma^*}|\mathcal{F}_t\right]=J_{L,t}^{x;\pi^*,c^*,\mu^*,\Sigma^*}$$
for any $0\leq t\leq s\leq T$. The rest of the proof is similar to
that of Theorem \ref{Theorem1 for power} and is thus omitted.
\end{proof}

\subsection{One-dimensional case}

In the rest of this paper, we focus on one-dimensional case, and
derive explicit solutions for the optimal investment-consumption
strategies and the worst-case parameters. Assume that $d=d'=1$ and
$\mathbb{A}=[\underline{\pi},\overline{\pi}]\times[\underline{c},\overline{c}]$,
where $\underline{\pi},\overline{\pi},\underline{c},\overline{c}$
are constants satisfying $-\infty\leq \underline{\pi}\leq0,\;1\leq
\overline{\pi}\leq+\infty,\;0\leq
\underline{c}\leq\overline{c}\leq+\infty$.

For $i=P,L$, we split the function $F_i$ (cf. (\ref{definition of
F})) into two parts as $F_i(x_q;x_\pi,x_c;x_\mu,{x_\sigma})=
f_i(x_q;x_c)+ g(x_\pi;x_\mu,{x_\sigma})$, where {we used the
notation $x_\sigma$ to replace $x_\Sigma$ in (\ref{definition of F})
to emphasize the one-dimensional setting, and $f_i, g$ are defined
as follows,} \be\label{definition of f}
 f_i(x_q;x_c):=\left\{
 \begin{array}{ll}
 {{\lambda\over p}e^{-x_q}x_c^{p}-x_c},&i=P;
 \vspace{2mm}\\
 \lambda e^{-x_q}\ln x_c- x_c,&i=L,
 \end{array}
 \right.
\ee
and
\be\label{definition of g}
 {g(x_\pi;x_\mu,{x_\sigma})}
 :={p-1\over 2}\,{x_\sigma}x_{\pi}^2
 +\left(x_{\mu}x_{\pi}+r(1-x_{\pi})-(R-r)(1-x_{\pi})^-\right).
\ee Herein, with a slight abuse of notation, we take $p=0$ in the
function $g$ for $i=L$.

It is clear that for any $x_q\in\mathbb{R}$,
$(x_{\pi}^*,\widetilde{x}^*_{c,i}(x_q);x_{\mu}^*,{x_\sigma^*})$ is a
saddle point of $F_i$ in $\mathbb{A}\times\mathbb{B}$, if
$\widetilde{x}^*_{c,i}(x_q)$ is the maximum point of
$f_i(x_q;\cdot)$ in the interval $[\underline{c},\overline{c}]$ and
$(x_{\pi}^*;x_{\mu}^*,{x_\sigma^*})$ is a saddle point of $g$ in
$[\underline{\pi},\overline{\pi}]\times \mathbb{B}$, i.e.
\be\label{optimalconsumption}
  f_i(x_q;\widetilde{x}_{c,i}^{*}(x_q))=\max_{x_c\in[\,\underline{c},\overline{c}\,]}f_i(x_q;x_c);
\\[2mm]\label{optimalportfolio}
g(x_{\pi};x_{\mu}^*,{x_\sigma^*})\leq
g(x_{\pi}^*;x_{\mu}^*,{x_\sigma^*})\leq
g(x_{\pi}^*;x_{\mu},{x_\sigma}) \ee for any
${(x_{\pi};x_{\mu},x_{\sigma})}\in[\underline{\pi},\overline{\pi}]\times
\mathbb{B}$.

From (\ref{definition of f}), it is immediate that the maximum value
and maximum point of $f_i$ in the interval
$[\underline{c},\overline{c}]$ take the form\be\label{optimalg1}
f_P(x_q;\widetilde{x}_{c,P}^*(x_q))=
 \left\{
 \begin{array}{lll}
 {\lambda\over p}\overline{c}^pe^{-x_q}-\overline{c},
 &\ \text{if}\ x_q<(p-1)\ln\overline{c}+\ln\lambda;
 \vspace{2mm}\\
 {(1-p)\lambda^{1/(1-p)}\over p}e^{x_q/(p-1)},
 &\ \text{if}\ (p-1)\ln\overline{c}+\ln\lambda\leq x_q\leq (p-1)\ln\underline{c}+\ln\lambda;
 \vspace{2mm}\\
 {\lambda\over p}\underline{c}^pe^{-x_q}-\underline{c},
 &\ \text{if}\ x_q>(p-1)\ln\underline{c}+\ln\lambda.
 \end{array}
 \right.
\ee
where
\begin{equation}\label{widehat_c}
 \widetilde{x}_{c,P}^*(x_q):=\underline{c}1_{\{\widehat{c}_P(x_q)\leq
 \underline{c}\}}+
 \widehat{c}_P(x_q)1_{\{\underline{c}<\widehat{c}_P(x_q)<
 \overline{c}\}}+ \overline{c}1_{\{\widehat{c}_P(x_q)\geq
 \overline{c}\}},\;\;
 \widehat{c}_P(x_q):=\lambda^{\frac{1}{1-p}}\exp\left(\frac{x_q}{p-1}\right),
\end{equation}
and \be\label{optimalg1 for log}
f_L(x_q;\widetilde{x}_{c,L}^*(x_q))=
 \left\{
 \begin{array}{lll}
 \lambda e^{-x_q}\ln \overline{c}- \overline{c},
 &\ \text{if}\ x_q<\ln\lambda-\ln\overline{c};
 \vspace{2mm}\\
 \lambda e^{-x_q}(\ln\lambda-x_q-1),
 &\ \text{if}\ \ln\lambda-\ln\overline{c}\leq x_q\leq \ln\lambda-\ln\underline{c};
 \vspace{2mm}\\
 \lambda e^{-x_q}\ln \underline{c}- \underline{c},
 &\ \text{if}\ x_q>\ln\lambda-\ln\underline{c}.
 \end{array}
 \right.
\ee
where
\begin{equation}\label{widehat_c for log}
 \widetilde{x}_{c,L}^*(x_q)=\underline{c}1_{\{\widehat{c}_L(x_q)\leq
 \underline{c}\}}+
 \widehat{c}_L(x_q)1_{\{\underline{c}<\widehat{c}_L(x_q)<
 \overline{c}\}}+ \overline{c}1_{\{\widehat{c}_L(x_q)\geq
 \overline{c}\}},\;\;
 \widehat{c}_L(x_q):=\lambda e^{-x_q}.
\end{equation}

In the case of power utility function, the corresponding ODE
(\ref{intergralequation}) has a financial interpretation. The
exponential of the ODE's solution $e^{q_P(t)}$ represents the
investor's extra utilities obtained by optimizing over all
admissible portfolio-consumption strategies (least affected by model
uncertainty) in the remaining horizon $[t,T]$, and in the
literature, $e^{q_P(t)}$ is dubbed as a (deterministic)
\emph{opportunity process} (see \cite{nutz}).

{Moreover, ODE (\ref{intergralequation}) and the definition of
$f_P(\cdot;\cdot)$ imply
$$
  \frac{-(e^{q_P(t)})^\prime}{e^{q_P(t)}}
  =-q_P^\prime(t)=p
  f_P(q_P(t);c^*(t))+pg(x_{\pi}^*;x_{\mu}^*,{x_\sigma^*})-\rho,
$$
where, with a slight abuse of notation, we denote
$$
 c^*(t):=\widetilde{x}_{c,P}^*(q_P(t)).
$$
Hence, we can further interpret ODE (\ref{intergralequation}) as a
description of} {the relative changing rate} of the opportunity
process $e^{q_P(t)}$, which {consists of} three factors: (i)
{\emph{the consumption contributing factor}
$pf_P(q_P(\cdot),c^*(\cdot))$, representing the change of the
opportunity process due to the consumption optimization, and
including two parts: current contribution $\lambda
e^{-q_P(\cdot)}(c^*(\cdot))^{p}/p$ and future contribution
$-c^*(\cdot)$}; (ii) \emph{the future investment contributing
factor} $pg(x_{\pi}^*;x_{\mu}^*,{x_\sigma^*})$, representing the
change of the opportunity process due to the portfolio optimization
in the remaining horizon; and (iii) \emph{the discount rate} $\rho$.
Increasing the consumption and future investment contributing
factors or decreasing the discount rate will lead to a larger
opportunity process.

{The current consumption contributing factor is the only one
affecting the instantanous utility, which is also reflected in the
expression of the expected utility (\ref{objectfunctional1}). The
future consumption contributing factor and the future investment
contributing factor determine the future consumption and terminal
utility through the channel of the future wealth. The player
achieves the maximum utility through balancing the risky asset and
riskless asset via the investment strategy, while balancing the
current utility and future utility via the consumption strategy.
Moreover, the definition of $f_P(q_P(\cdot),c^*(\cdot))$ implies
that $\lambda e^{-q_P(\cdot)}$ is the weight of the current
consumption utility relative to the future utility, which is
consistent with our intuition that increasing opportunity process
will lead to a larger weight of the future utility, and decrease the
current consumption.}

\section{Explicit solutions of the optimal strategies and worst-case parameters}

\subsection{The worst-case parameters and the optimal portfolios}

In this section, we further compute the saddle point
$(x_{\pi}^*;x_{\mu}^*,{x_\sigma^*})$ of the function
$g(\cdot;\cdot,\cdot)$ given in (\ref{definition of g}). It then
follows from Theorems \ref{Theorem1 for power} and \ref{Theorem1 for
log} that the saddle point provides an explicit solution for the
worst-case parameters and the optimal portfolio of the maxmin
problem (\ref{problem}) by letting
$(\mu^*_s,\sigma^*_s)=(x_{\mu}^{*},{\sqrt{x_{\sigma}^*}})$ and
$\pi^*_s=x_{\pi}^*$ for $s\in[0,T]$. In particular, we consider two
specific examples of the uncertain parameter sets $\mathbb{B}$.
\smallskip

\begin{assumption}\label{no correlation}
Assume that $\mathbb{B}=[\,\underline{\mu},\overline{\mu}\,]\times
[\,\underline{\sigma}^2,\overline{\sigma}^2\,]$, where
$\underline{\mu},
\overline{\mu},\underline{\sigma},\overline{\sigma}$ are constants
satisfying $-\infty<\underline{\mu}\leq \overline{\mu}<+\infty,\,
0\leq\underline{\sigma}\leq\overline{\sigma}<+\infty$ and
$\overline{\sigma}>0$.
\end{assumption}

\begin{theorem}\label{theorem_portfolio}
Under Assumption \ref{no correlation}, the worst-case parameters
$(\mu^*,\sigma^*)$ and the optimal portfolio $\pi^*$ are given as
follows:

(i) the worst-case drift and volatility are
$$(\mu^*_s,\sigma_s^*)=\left(\underline{\mu}1_{\{\underline{\mu}>r\}}+
[\underline{\mu},\overline{\mu}]1_{\{\underline{\mu}\leq r\leq
\overline{\mu}\}}+\overline{\mu}1_{\{\overline{\mu}<r\}},\overline{\sigma}\right)
$$ for $s\in[0,T]$, where $[\underline{\mu},\overline{\mu}]$ means
$\mu^*_s$ may take any value in that interval;

(ii) the optimal portfolio is a constant process, which is
summarized in Table 1, with $\beta_1$, $\beta_2$ and $\beta_3$ given
as
 $$
  \beta^1{:=}{\underline{\mu}-R\over (1-p)\overline{\sigma}^2},\quad
  \beta^2{:=}{\underline{\mu}-r\over (1-p)\overline{\sigma}^2},\quad
  \beta^3{:=}{\overline{\mu}-r\over (1-p)\overline{\sigma}^2}.
  $$

\begin{table}[h!]
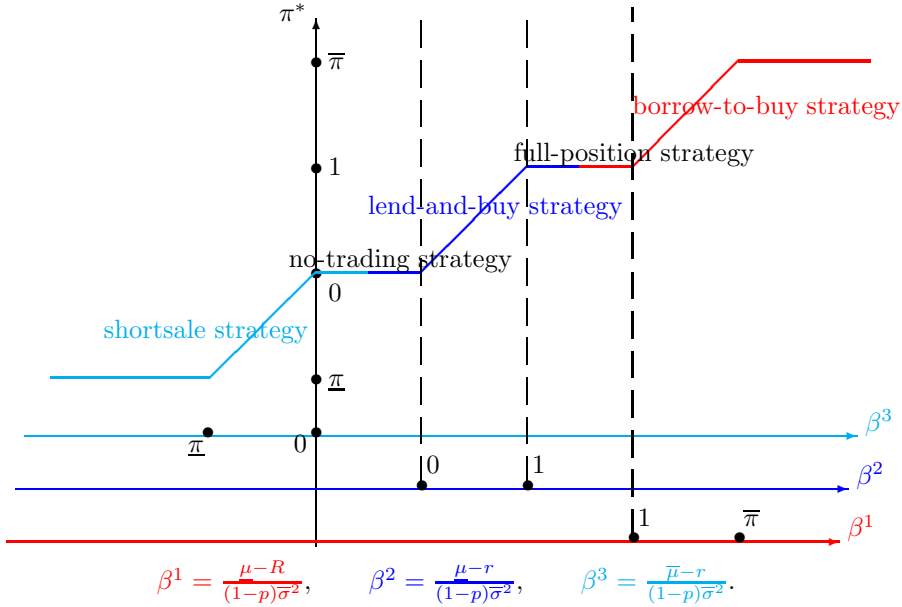
    \caption{the optimal portfolio strategies}
    \ \ \ \ \ \ \ \ \ \ \ \
    \begin{tabular}{c c c c c c}
      \hline
      & $\beta^1\geq 1$ &$\beta^1\leq1\leq\beta^2$ & $0\leq\beta^2\leq1$
      & $\beta^2\leq 0\leq \beta^3$& $\beta^3\leq 0$ \\
      \hline
      $\;\;\pi^*_s$ &$\min\{\beta^1,\overline{\pi}\}$ & $1$ & $\beta^2$ & $0$ & $\max\{\beta^3,\underline{\pi}\}$\\
      \hline
    \end{tabular}  \label{table1}
  \end{table}
\end{theorem}

\begin{proof} Due to its length, the proof is postponed to Appendix A.
\end{proof}

We note that the worst-case volatility $\sigma^*$ attains its upper
bound $\overline{\sigma}$. This is due to the fact that the value
function is monotone in volatility ${\sigma}$ in the one-dimensional
setting. A larger $\sigma$ means the investor faces more market
risks, and therefore, she will have a smaller value function.

On the other hand, the worst-case drift is a bang-bang type. By the
assertion (ii) about the optimal portfolio strategies, we know that
$\underline{\mu}>r$ implies $\pi^*>0$, i.e. the investor holds a
long position of the stock. The worst-case drift is therefore its
lower bound. Likewise, $\overline{\mu}<r$ implies $\pi^*<0$, and
therefore, the worst-case drift takes its upper bound. If
$\underline{\mu}\leq r\leq \overline{\mu}$, then $\pi^*\equiv 0$, so
the estimation of the drift is irrelevant in this situation.

From Table 1, we categorize five different optimal portfolio
strategies $\pi^*$ according to various scenarios.

(i) \emph{Borrow-to-buy strategy.} When $\beta^1\geq 1$, the
investor will borrow $(\min\{\beta^1,\overline{\pi}\}-1)$ units of
her wealth with borrowing rate $R$ to invest in the stock, and the
optimal portfolio is $\min\{\beta^1,\overline{\pi}\}$. The reason is
that in this situation, $\underline{\mu}\geq
R+(1-p)\overline{\sigma}^2$, i.e. the stock's return even with the
worst estimation of the drift is still higher than the borrowing
cost. Hence, the stock's high risk premium attracts the investor to
{borrow to} invest as much as possible {to approach the optimal
strategy without constraint}.

(ii) \emph{Full-position strategy.} When $\beta^1\leq 1\leq
\beta^2$, the investor will simply invest all her wealth in the
stock with no additional borrowing or lending. In this case, since
$\underline{\mu}\leq R+(1-p)\overline{\sigma}^2$, there exists a
possibility that the stock's return may not be good enough to
compensate for the borrowing cost. As a result, the investor would
prefer not to borrow. On the other hand, since $\underline{\mu}\geq
r+(1-p)\overline{\sigma}^2$, the stock's return even in the worst
scenario is still better than the return from the bank account, and
accordingly, the investor would put all her money in the stock
rather than in the bank account.

(iii) \emph{Lend-and-buy strategy}. When $0\leq \beta^2\leq 1$, the
investor will invest $\beta^2$ proportion of her wealth in the
stock, and the remaining proportion $(1-\beta^2)$ in the bank
account to earn the interest rate $r$. This is similar to the
standard Merton's strategy with Sharpe ratio
$(\underline{\mu}-r)/\overline{\sigma}$.

(iv) \emph{No-trading strategy}. When $\beta^2\leq 0\leq \beta^3$,
the investor will put all her money in the bank account. In this
case, $\underline{\mu}\leq r\leq \overline{\mu}$, so there is a risk
that the return from buying the stock is not as good as holding the
bank account, and the investor would prefer not to invest in the
stock. On the other hand, the best estimation of the drift
$\overline{\mu}$ is still better than the interest rate $r$, so
implementing a shortsale strategy may incur a potential loss for the
investor. This refrains her from short selling the stock.

(v) \emph{Shortsale strategy}. When $\beta^3\leq 0$, the investor
will hold short position in the stock as much as possible, which is
$\max\{\beta_3,\underline{\pi}\}$ units of her wealth in this
situation. Consequently, she keeps
$(1-\max\{\beta_3,\underline{\pi}\})$ units of her wealth in the
bank account in order to earn the interest rate $r$.

We can further illustrate the above five optimal portfolio
strategies via the following figure, where the horizontal axes
represent the values of $\beta^1$, $\beta^2$ and $\beta^3$ from the
top to the bottom, and the vertical axis represents the optimal
portfolio.\\

\begin{picture}(0,230)(190,-15)\label{figureinvestment}
\put(307,6){\vector(0,1){200}}\put(293,205){$\pi^*$}
\textcolor{red}{\put(190,8){\vector(1,0){315}}\put(508,10){$\beta^1$}}
\textcolor{blue}{\put(190,28){\vector(1,0){315}}\put(508,30){$\beta^2$}}
\textcolor{cyan}{\put(190,48){\vector(1,0){315}}\put(508,50){$\beta^3$}}
\put(257,47){$\bullet$}\put(252,42){$\underline{\pi}$}
\put(298,47){$\bullet$}\put(292,42){$0$}
\put(338,27){$\bullet$}\put(342,34){$0$}
\put(378,27){$\bullet$}\put(382,34){$1$}
\put(418,7){$\bullet$}\put(422,14){$1$}
\put(458,7){$\bullet$}\put(462,14){$\overline\pi$}
\put(298,67){$\bullet$}\put(305,67){$\underline{\pi}$}
\put(298,107){$\bullet$}\put(305,99){$0$}
\put(298,147){$\bullet$}\put(305,147){$1$}
\put(298,187){$\bullet$}\put(305,187){$\overline\pi$}
{\thicklines\textcolor{cyan}{\put(200,70){\line(1,0){60}}\put(260,70){\line(1,1){40}}\put(300,110){\line(1,0){20}}}\textcolor{blue}{\put(320,110){\line(1,0){20}}
\put(340,110){\line(1,1){40}}\put(380,150){\line(1,0){20}}}\textcolor{red}{\put(400,150){\line(1,0){20}}\put(420,150){\line(1,1){40}}
\put(460,190){\line(1,0){50}}}}
\put(340,30){\line(0,1){10}}\put(340,45){\line(0,1){10}}\put(340,60){\line(0,1){10}}
\put(340,75){\line(0,1){10}}\put(340,90){\line(0,1){10}}\put(340,105){\line(0,1){10}}
\put(340,120){\line(0,1){10}}\put(340,135){\line(0,1){10}}\put(340,150){\line(0,1){10}}
\put(340,165){\line(0,1){10}}\put(340,180){\line(0,1){10}}\put(340,195){\line(0,1){10}}
\put(380,30){\line(0,1){10}}\put(380,45){\line(0,1){10}}\put(380,60){\line(0,1){10}}
\put(380,75){\line(0,1){10}}\put(380,90){\line(0,1){10}}\put(380,105){\line(0,1){10}}
\put(380,120){\line(0,1){10}}\put(380,135){\line(0,1){10}}\put(380,150){\line(0,1){10}}
\put(380,165){\line(0,1){10}}\put(380,180){\line(0,1){10}}\put(380,195){\line(0,1){10}}
\put(420,10){\line(0,1){10}}\put(420,25){\line(0,1){10}}\put(420,40){\line(0,1){10}}
\put(420,55){\line(0,1){10}}\put(420,70){\line(0,1){10}}\put(420,85){\line(0,1){10}}
\put(420,100){\line(0,1){10}}\put(420,115){\line(0,1){10}}\put(420,130){\line(0,1){10}}
\put(420,145){\line(0,1){10}}\put(420,160){\line(0,1){10}}\put(420,175){\line(0,1){10}}
\put(420,190){\line(0,1){10}}\put(420,205){\line(0,1){5}}
\put(220,85){\textcolor{cyan}{shortsale
strategy}}\put(290,112){no-trading
strategy}\put(320,132){\textcolor{blue}{lend-and-buy strategy}}
\put(375,152){full-position
strategy}\put(420,170){\textcolor{red}{borrow-to-buy strategy}}
\put(240,-10){$\textcolor{red}{\beta^1={\underline{\mu}-R\over
(1-p)\overline{\sigma}^2}},\qquad
  \textcolor{blue}{\beta^2={\underline{\mu}-r\over (1-p)\overline{\sigma}^2}},\qquad
  \textcolor{cyan}{\beta^3={\overline{\mu}-r\over
  (1-p)\overline{\sigma}^2}}.$}
\end{picture}
\begin{center}
Figure 1: the optimal portfolio strategies
\end{center}

In the existing literature, the worst-case parameters are usually
bang-bang type, i.e., they take values at the boundaries of the
uncertain parameter set. Next, we give an example where the
worst-case drift and volatility are an interior point in the
uncertain parameter set. In particular, the worst-case volatility
may not be its upper bound anymore.
\smallskip

\begin{assumption}\label{correlation}
Assume that $R=r,\overline{\pi}=+\infty,\underline{\pi}=-\infty$ and
$\mathbb{B}=\{(\mu,\sigma):\mu=\underline{\mu}+\alpha,\sigma=\underline{\sigma}^2+k\alpha^q,\alpha\in[\,0,\overline{\alpha}\,]\,\}$,
where $\underline{\mu},\underline{\sigma},k,q,\overline{\alpha}$ are
constants satisfying {$\underline{\sigma}\geq0,
k>0,0<q<1,\overline{\alpha}\geq0$.}\smallskip
\end{assumption}

The set $\mathbb{B}$ indicates that the ambiguities about drift and
volatility are correlated. A higher return is associated with a
larger risk. The limiting case $q=1$ means that the relationship
between the ambiguity about drift and the ambiguity about the
volatility square is linear, which is just Example 2.4 in
\cite{Epstein}. The other spectrum $q=0$ means no ambiguity about
volatility. Finally, $0<q<1$ means that the relationship between the
ambiguity about drift and the ambiguity about {the volatility
square} is sub-linear.

\begin{theorem}\label{theorem for ambiguity correlation}
Under Assumption \ref{correlation}, the worst-case parameters
$(\mu^*,\sigma^* )$ and the optimal portfolio $\pi^*$ are given as
follows:

(i) the worst-case parameters
$(\mu^*,\sigma^*)=(\underline{\mu}+\alpha^*,\sqrt{\underline{\sigma}^2+k(\alpha^*)^q})$,
with
$$
\alpha^*=\left\{
 \begin{array}{ll}
 r-\underline{\mu},\qquad &-\overline{\alpha}<\underline{\mu}-r\leq 0;
 \\[2mm]
 \alpha_0, &0<\underline{\mu}-r
 <\widehat{\alpha}:=[\,2\underline{\sigma}^2\overline{\alpha}^{1-q}+{k(2-q)\overline{\alpha}\,]/(kq)};
 \\[2mm]
 \overline{\alpha},&\mbox{otherwise},
 \end{array}
 \right.
$$
where $\alpha_0$ is the solution of the following algebra equation
(for the case $\underline{\mu}-r>0$), \be\label{algebra equation}
 h_1(\alpha):=2\underline{\sigma}^2+{k(2-q)\alpha^q-k q(\underline{\mu}-r)\alpha^{q-1}}=0;
\ee

(ii) the optimal portfolio $\pi^*$ is a constant process given by
$(\mu^*-r)/((1-p)(\sigma^*)^2)$.
\end{theorem}
\begin{proof}
First, we prove that the algebra equation~\eqref{algebra equation}
has a unique zero crossing point $\alpha_0$ for the case
$\underline{\mu}-r>0$ and, moreover,
$\alpha_0\in(0,\overline{\alpha})$ if $0<\underline{\mu}-r
<\widehat{\alpha}$. Indeed, it is not difficult to check that for
$\underline{\mu}-r>0$, we have
$$
 \lim\limits_{\alpha\rightarrow0^+}h_1(\alpha)=-\infty,\;\; \lim\limits_{\alpha\rightarrow+\infty}h_1(\alpha)=+\infty,
$$
and
$$
 h^\prime_1(\alpha)={kq(2-q)\alpha^{q-1}+kq(1-q)(\underline{\mu}-r)\alpha^{q-2}}>0.
$$
Hence, $h_1(\cdot)$ in the algebra equation~\eqref{algebra equation}
has a unique zero crossing point $\alpha_0$ for
$\underline{\mu}-r>0$. Moreover, direct computations show that if
$\underline{\mu}-r<\widehat{\alpha}$, then
$$
  h_1(\overline{\alpha})=2\underline{\sigma}^2+{k(2-q)\overline{\alpha}^q-kq(\underline{\mu}-r)\overline{\alpha}^{q-1}}>0,
$$
which means $\alpha_0\in(0,\overline{\alpha})$ if
$0<\underline{\mu}-r <\widehat{\alpha}$.

Secondly, we show that
$$
  g(\pi^*;\mu^*,(\sigma^*)^2)\leq g(\pi^*;x_\mu,{x_\sigma}),\;\;(x_\mu,{x_\sigma})\in\mathbb{B}.
$$
To this end, note that \bee
  g(\pi^*;x_\mu,{x_\sigma})&=&{p-1\over 2}\,(\,\underline{\sigma}^2+k\alpha^q\,) (\pi^*)^2+(\underline{\mu}+\alpha-r)\pi^*+r
  ={(\underline{\mu}+\alpha^*-r)h_2(\alpha)\over 2{(1-p)}(\sigma^*)^4}+r,
\eee
where
$$
 h_2(\alpha):={-(\underline{\sigma}^2+k\alpha^q)}(\underline{\mu}+\alpha^*-r)
 +2[\,\underline{\sigma}^2+k(\alpha^*)^q\,](\underline{\mu}+\alpha-r),
$$
and
$$
 h_2^\prime(\alpha)={-kq}\alpha^{q-1}(\underline{\mu}+\alpha^*-r)+2[\,\underline{\sigma}^2+k(\alpha^*)^q\,],\;\;
 h_2^{\prime\prime}(\alpha)={kq(1-q)}\alpha^{q-2}(\underline{\mu}+\alpha^*-r).
$$

We divide the possible values of $\underline{\mu}-r$ into four
cases. Fix $\alpha\in[\,0,\overline{\alpha}\,]$ and
$(x_\mu,x_\sigma)\in\mathbb{B}$. If
$\underline{\mu}-r\leq-\overline{\alpha}$, then
$$ \underline{\mu}+\alpha^*-r\leq 0,\;\;h_2^\prime(\alpha)>0,\;\;h_2(\alpha)\leq h_2(\overline{\alpha})=h_2(\alpha^*),\;\;
 g(\pi^*;x_\mu,{x_\sigma})\geq  g(\pi^*;\mu^*,(\sigma^*)^2).
$$
If $-\overline{\alpha}<\underline{\mu}-r\leq 0$, then
$$
 \underline{\mu}+\alpha^*-r=\underline{\mu}+(r-\underline{\mu})-r=0,\;\;\pi^*=0,\;\;
 g(\pi^*;x_\mu,{x_\sigma})=r=g(\pi^*;\mu^*,(\sigma^*)^2).
$$
If $0<\underline{\mu}-r<\widehat{\alpha}$, then
$$
 \underline{\mu}+\alpha^*-r>0,\;\;h_2^{\prime\prime}(\alpha)>0,\;\;
 h_2^{\prime}(\alpha^*)=0,\;\;h_2(\alpha)\geq h_2(\alpha^*),\;\;
 g(\pi^*;x_\mu,{x_\sigma})\geq  g(\pi^*;\mu^*,(\sigma^*)^2),
$$
where we have used the fact that $h_1(\alpha^*)=0$ implies that
$h_2^{\prime}(\alpha^*)=0$. Finally, if $\underline{\mu}-r\geq
\widehat{\alpha}$, then
$$
 h_2^{\prime}(\overline{\alpha})=h_1(\overline{\alpha})\leq0,\;\;h_2^{\prime\prime}(\alpha)>0,\;\;h_2^{\prime}(\alpha)\leq0,\;\;
 h_2(\alpha)\geq  h_2(\alpha^*),\;\;
 g(\pi^*;x_\mu,{x_\sigma})\geq g(\pi^*;\mu^*,(\sigma^*)^2).
$$

Thirdly, we prove that
$$
  g(\pi^*;\mu^*,(\sigma^*)^2)\geq g(x_\pi;\mu^*,(\sigma^*)^2),\;\;x_\pi\in{\mathbb{R}}.
$$
To see this, we note that \bee
  g(x_\pi;\mu^*,(\sigma^*)^2)&=&{p-1\over 2}(\sigma^*)^2 x_\pi^2+(\mu^*-r)x_\pi+r
  \\[2mm]
  &=&{p-1\over 2}(\sigma^*)^2\left[\,x_\pi-{\mu^*-r\over (1-p)(\sigma^*)^2}\,\right]^2
  +{(\mu^*-r)^2\over 2(1-p)(\sigma^*)^2}+r.
\eee It is then clear that $g(x_\pi;\mu^*,(\sigma^*)^2)$ attains its
maximum at the point $x_\pi=\pi^*$, so $(x_\pi^*
;\mu^*,(\sigma^*)^2)$ is a saddle point of $g$, and the conclusion
follows from Theorems \ref{Theorem1 for power} and \ref{Theorem1 for
log}.
\end{proof}

If $\underline{\mu}-r\leq -\overline{\alpha}$, then the investor
will short sell her stock, which is similar to {\em Shortsale
strategy} in {Theorem \ref{theorem_portfolio}}. Moreover, $\pi^*<0$
implies that the worst-case drift $\mu^*$ and volatility $\sigma^*$
attain their respective upper bounds
$\underline{\mu}+\overline{\alpha}$ and
$\sqrt{\underline{\sigma}^2+k(\overline{\alpha})^q}$.

If $-\overline{\alpha}<\underline{\mu}-r{< 0}$, then the lower bound
of the drift $\underline{\mu}\leq r$ and the upper bound
$\underline{\mu}+\overline{\alpha}>r$. Similar to {\em No-trading
strategy} in {Theorem \ref{theorem_portfolio}}, the investor may
suffer losses if she buys or short sells the stocks, so she will
simply invest all her money in the bank account. Moreover, $\pi^*=0$
implies that the estimation of the drift and volatility is
irrelevant in this situation and, without loss of generality, we let
$(\mu^*,\sigma^*)=\left(\,r,\sqrt{\underline{\sigma}^2+k(r-\underline{\mu})^q}\,\right)$.

If $\underline{\mu}-r>0$, then the investor will invest in the stock
according to the optimal proportion
$\pi^*=(\mu^*-r)/((1-p)(\sigma^*)^2)>0$. {If there is no ambiguity
about volatility, the worst-case drift is its lower bound
$\underline{\mu}$ and $\alpha^*=0$. Since the correlation between
the uncertain drift and uncertain volatility is positive, the
worst-case parameter $\alpha^*=\alpha_0$, which is an interior point
of the interval $[\,0,\overline{\alpha}\,]$ if $0<\underline{\mu}-r
 <\widehat{\alpha}$.} In
particular, the worst-case volatility may not be its upper bound
anymore. This is in contrast to Theorem
 \ref{theorem_portfolio}, where the worst-case
parameters take values at the boundaries of the uncertain parameter
set.

\subsection{The Optimal consumption under power utility}

In this section, we compute the explicit solutions to ODEs
(\ref{intergralequation}) and (\ref{intergralequation1 for log}),
which in turn allows us to construct the optimal consumption of the
maxmin problem (\ref{problem}) (cf. (\ref{widehat_c}) or
(\ref{widehat_c for log})). Note that if $\lambda=0$ in
(\ref{objectfunctional1}) and $\overline{c}=\underline{c}$, the
consumption does not play a role and the optimal consumption
strategy is simply $c^*_t=\overline{c}=\underline{c}$. Hence, we
focus on the case $\lambda>0$ and $\overline{c}>\underline{c}$ in
the rest of the paper. We first present the result for power
utility.

\begin{theorem}\label{theorem_consumption}
Let $T>0$ be {a large enough} number. For the power utility case,
the optimal consumption $c^*_t=c^*(t)$, $t\in[0,T]$, is a
deterministic process, where $c^*(t)=\widetilde{x}_{c,P}^*(q_P(t))$
with $\widetilde{x}_{c,P}^*(\cdot)$ given in (\ref{widehat_c}) {and
$q_P(\cdot)$ given in Table 8}. Moreover, the optimal consumption
$c^*_t$ is summarized in Table 2\footnote{{In the case of
$\underline{c}=0$, the results are similar to those in Table
\ref{table2} except that
$\underline{c}I_0^{123}+\widehat{c}(t)I^{12}_{123}+\overline{c}I^4_{12}$
and $\underline{c}I_0^{23}+\widehat{c}(t)I^4_{23}$ are replaced by
$\widehat{c}(t)I^{12}_0+\overline{c}I^4_{12}$ and $\widehat{c}(t)$,
respectively.} Note that when $\underline{c}=0$, since $\lambda>0$,
the last two rows about the optimal consumption are then
irrelevant.}.

  \begin{table}[h!]
    \caption{the optimal consumption {in the case of $\underline{c}>0$}}
    \begin{tabular}{c c c c c c }
       \hline
        &  {\tiny$\rho-pK$} & & & &   \\
       & {\tiny$(-\infty,(1-p)\underline{c}\,)$} &  {\tiny$\{(1-p)\underline{c}\,\}$} &
       {\tiny$((1-p)\underline{c},(1-p)\overline{c}\,)$} &  {\tiny$\{(1-p)\overline{c}\,\}$}
       &  {\tiny$((1-p)\overline{c},+\infty)$}   \\
       \hline
       {\tiny$\underline{c}<\overline{c}<\lambda^{1/(1-p)}$} &  {\tiny$\underline{c}I_0^{123}+\widehat{c}(t)I^{12}_{123}+\overline{c}I^4_{12}$}
       &{\tiny$\widehat{c}(t)I^{12}_0+\overline{c}I^4_{12}$} &  {\tiny$\widehat{c}(t)I^{12}_0+\overline{c}I^4_{12}$} &  \cellcolor{yellow!25}{\tiny$\overline{c}$}
       &\cellcolor{yellow!25}{\tiny$\overline{c}$}
       \\
       {\tiny$\underline{c}<\overline{c}=\lambda^{1/(1-p)}$} & {\tiny$\underline{c}I_0^{23}+\widehat{c}(t)I^4_{23}$}
       &\cellcolor{red!25}{\tiny$\widehat{c}(t)$} & \cellcolor{red!25}{\tiny$\widehat{c}(t)$} & \cellcolor{yellow!25}{\tiny$\overline{c}$} & \cellcolor{yellow!25}{\tiny$\overline{c}$}
       \\
       {\tiny$\underline{c}<\lambda^{1/(1-p)}<\overline{c}$} & {\tiny$\underline{c}I_0^{23}+\widehat{c}(t)I^4_{23}$}
       & \cellcolor{red!25}{\tiny$\widehat{c}(t)$} & \cellcolor{red!25}{\tiny$\widehat{c}(t)$} & \cellcolor{red!25}{\tiny$\widehat{c}(t)$} & {\tiny$\overline{c}I^{21}_0+\widehat{c}(t)I^4_{21}$}
       \\
       {\tiny$\lambda^{1/(1-p)}=\underline{c}<\overline{c}$} & \cellcolor{blue!25}{\tiny$\underline{c}$} & \cellcolor{blue!25}{\tiny$\underline{c}$}
       & \cellcolor{red!25}{\tiny$\widehat{c}(t)$} & \cellcolor{red!25}{\tiny$\widehat{c}(t)$} &{\tiny$\overline{c}I^{21}_0+\widehat{c}(t)I^4_{21}$}
       \\
       {\tiny$\lambda^{1/(1-p)}<\underline{c}<\overline{c}$} & \cellcolor{blue!25}{\tiny$\underline{c}$} & \cellcolor{blue!25}{\tiny$\underline{c}$}
       & {\tiny$\widehat{c}(t)I_0^{32}+\underline{c}I^4_{32}$} & {\tiny$\widehat{c}(t)I_0^{32}+\underline{c}I^4_{32}$}
       & {\tiny$\overline{c}I^{321}_0+\widehat{c}(t)I^{32}_{321}+\underline{c}I^4_{32}$}
       \\
       \hline
    \end{tabular}  \label{table2}
  \end{table}
\noindent Herein, the constant $K$ in the table corresponds to the
future investment contributing factor in (\ref{intergralequation}),
and has the explicit form \be\label{constantK}
 K:=g(x_{\pi}^*;x_{\mu}^*,x_{\sigma}^*)=
 \left\{
 \begin{array}{ll}
 R+\overline{\pi}(\underline{\mu}-R)
 -{1-p\over 2}\,\overline{\sigma}^2\overline{\pi}^2,\quad
 &\beta^1\geq \overline{\pi};
 \vspace{2mm}\\
 R+{(\underline{\mu}-R)^2\over 2(1-p)\overline{\sigma}^2},\quad
 &1\leq\beta^1\leq\overline{\pi};
 \vspace{2mm}\\
 \underline{\mu}-{1-p\over2}\,\overline{\sigma}^2,\quad
 &\beta^1\leq1\leq\beta^2;
 \vspace{2mm}\\
 r+{(\underline{\mu}-r)^2\over 2(1-p)\overline{\sigma}^2},\quad
 &0\leq\beta^2\leq1;
 \vspace{2mm}\\
 r,\quad
 &\beta^2\leq0\leq \beta^3;
 \vspace{2mm}\\
 r+{(\overline{\mu}-r)^2\over 2(1-p)\overline{\sigma}^2},\quad
 &\underline{\pi}\leq\beta^3\leq0;
 \vspace{2mm}\\
 r+\underline{\pi}(\overline{\mu}-r)-{1-p\over 2}\,\overline{\sigma}^2\underline{\pi}^2,\quad
 &\beta^3\leq\underline{\pi},
 \end{array}
 \right.
\ee and {$\widehat{c}(t)=\widehat{c}_P(q_P(t))$} (cf.
(\ref{widehat_c})). The indicator function $I_{a}^{b}$ represents
the time period $[T_a,T_b]$ with $T_0=0$ and $T_4=T$, where the
explicit forms of different time periods are given in Appendix B.
\end{theorem}

\begin{proof} Due to its length, the proof is postponed to Appendix B.
\end{proof}

Table 2 lists all the possible consumption patterns under different
parameters. For example,
$\underline{c}I_0^{123}+\widehat{c}(t)I^{12}_{123}+\overline{c}I^4_{12}$
in the first row and the first column (left-top corner) is the
optimal consumption when the market parameters satisfy
$\underline{c}<\overline{c}<\lambda^{1/(1-p)}$ and
$\rho-pK\in(-\infty,(1-p)\underline{c}\,)$. More specifically, in
the time interval $[0,T_{123}]$, the investor will consume at the
minimum rate $\underline{c}$. Then the investor will consume at the
optimal rate $\widehat{c}(t)=\lambda^{1/(1-p)}{\exp(q_P(t)/(p-1))}$
in the time interval $[T_{123},T_{12}]$, since in this case
$\underline{c}\leq \widehat{c}(t)\leq \overline{c}$. Finally, in the
remaining time interval $[T_{12},T]$, the investor will consume at
the maximum rate $\overline{c}$.

In contrast, in the right-bottom corner, we obtain a reversed
consumption pattern when
$\lambda^{1/(1-p)}<\underline{c}<\overline{c}$ and
$\rho-pK\in((1-p)\overline{c},+\infty)$. That is, the consumption
will be decreasing from the maximum rate $\overline{c}$ in
$[0,T_{321}]$, to $\widehat{c}(t)$ in $[T_{321},T_{32}]$, and
finally to the minimal rate $\underline{c}$ in $[T_{32},T]$.

In the following, we give some intuitive explanations of different
consumption patterns. From the expression of $f_P$ and $c^*(t)$, we
know that the optimal consumption $c^*_t=c^*(t)$ achieves the
maximum of the concave function $f_P(q_P(t),\cdot)$ in the interval
$[\underline{c},\overline{c}]$. Moreover, note that
$\widehat{c}(t)=\lambda^{1/(1-p)}\exp(q_P(t)/(p-1))$ as in
(\ref{widehat_c}) is the maximum point of $f_P(q_P(t),\cdot)$ on
$\mathbb{R}_+$. Hence, $c^*_t=\widehat{c}(t)$ if
$\underline{c}<\widehat{c}(t)<\overline{c}$. Otherwise, $c^*_t$ will
be either $\underline{c}$ or $\overline{c}$.

From the proof of Proposition \ref{proposition} below, we know
$q_P(t)$ is monotone in time $t$, so is $\widehat{c}(t)$. As a
result, whether $\widehat{c}(t)$ stays in
$[\underline{c},\overline{c}]$ or not only depends on its values at
the two end points $\widehat{c}(T)$ and $\widehat{c}(0)$, and their
relationship with $\underline{c}$ and $\overline{c}$.

In fact, it follows from $q_P(T)=0$ that
$\widehat{c}(T)=1/\lambda^{1-p}$. By the continuity of
$\widehat{c}(t)$, when $t$ approaches maturity $T$, $c^*(t)$ will
reach its upper bound $\overline{c}$ if
$\underline{c}<\overline{c}<\lambda^{1/(1-p)}$; $c^*(t)$ will be
precisely $\widehat{c}(t)$ if
$\underline{c}<\lambda^{1/(1-p)}<\overline{c}$; $c^*(t)$ will reach
its lower bound $\underline{c}$ if
$\lambda^{1/(1-p)}<\underline{c}<\overline{c}$. The above three
situations thus determine the classification of the rows in Table 2.

On the other hand, we have the following asymptotic results for
$\lim\limits_{T\rightarrow+\infty}\widehat{c}(0)$ in Table 3 ({see
also Appendix B, in particular
(\ref{Asymptoticproperty1})-(\ref{Asymptoticproperty4})}). By the
continuity of $\widehat{c}(t)$, when $T$ is large enough and $t$ is
near initial time $0$, $c^*(t)=\underline{c}$ if
$\rho-pK\in(-\infty,(1-p)\underline{c}\,)$; $c^*(t)=\widehat{c}$ if
$\rho-pK\in((1-p)\underline{c},(1-p)\overline{c}\,)$; and
{$c^*(t)=\overline{c}$} if $\rho-pK\in((1-p)\overline{c},+\infty)$.
Consequently, the above three situations divide the columns in Table
2.

  \begin{table}[h!]
    \caption{the limit of $\widehat{c}(0)$ when $T\rightarrow+\infty$}
    \begin{tabular}{c c c c c c }
       \hline
       $\rho-pK$ & $(-\infty,(1-p)\underline{c}\,)$ &  $\{(1-p)\underline{c}\,\}$ &
       $((1-p)\underline{c},(1-p)\overline{c}\,)$ &  $\{(1-p)\overline{c}\,\}$
       & $((1-p)\overline{c},+\infty)$    \\
       \hline
       $\lim\limits_{T\rightarrow+\infty}\,\widehat{c}(0)$ & $<\underline{c}$ & $=\underline{c}$
       & $\in(\,\underline{c},\overline{c}\,)$ & $=\overline{c}$
       & $>\overline{c}$  \\
       \hline
    \end{tabular}  \label{table3}
  \end{table}

Next, we further show that the optimal consumption admits some time
monotone properties. As opposed to the unconstrained consumption
case, the consumption constraints may force the optimal consumption
to be either nonincreasing or nondecreasing no matter the value of
$(\rho-pK)$.

\begin{proposition}\label{proposition}
The optimal consumption $c^*_t$, $t\in[0,T]$, has the following
monotone properties in time $t$, as specified in Table 4. The
symbols $\nearrow,\searrow$ and $\perp$ represent nondecreasing,
nonincreasing and independent of time $t$, respectively.
 \begin{table}[h!]
  \centering
  \fontsize{10}{12}\selectfont
  \caption{the optimal consumption in time} \label{table42}
    \begin{tabular}{|c|c|c|c|c|c|c}
    \hline
    {\scriptsize$\underline{c}<\overline{c}<\lambda^{1/(1-p)}$}&\multicolumn{6}{c|}{{\scriptsize$\nearrow$}}\cr\cline{2-7}
    \hline
    \multirow{2}{*}{{\scriptsize$\underline{c}\leq \lambda^{1/(1-p)}\leq \overline{c}$}}&
    \multicolumn{2}{c|}{{\scriptsize$\rho-pK<(1-p)\lambda^{1/(1-p)}$}}&\multicolumn{2}{c|}{{\scriptsize$\rho- pK=(1-p)\lambda^{1/(1-p)}$}}
    &\multicolumn{2}{c|}{{\scriptsize$\rho-pK>(1-p)\lambda^{1/(1-p)}$}}\cr\cline{2-7}
    &\multicolumn{2}{c|}{{\scriptsize$\nearrow$}}&\multicolumn{2}{c|}{{\scriptsize$\perp$}}&
    \multicolumn{2}{c|}{{\scriptsize$\searrow$}}\cr\cline{2-7}
    \hline
    {\scriptsize$\lambda^{1/(1-p)}<\underline{c}<\overline{c}$}&\multicolumn{6}{c|}{{\scriptsize$\searrow$}}\cr\cline{2-7}
    \hline
    \end{tabular}
\end{table}

\end{proposition}

\begin{proof} It follows from the expressions of $\widetilde{x}_{c,P}^*(x_q)$ and $\widehat{c}_P(x_q)$ in $(\ref{widehat_c})$ that if $q_P(t)$ is
nonincreasing, then {$c^*(t)=\widetilde{x}_{c,P}^*(q_P(t))$} is
nondecreasing; if $q_P(t)$ is nondecreasing, then $c^*(t)$ is
nonincreasing. On the other hand, The expression {(\ref{optimalg1})}
and ODE (\ref{intergralequation}) lead to
\begin{equation}\label{second_derivative}
q_P^{\prime\prime}(t)={-p\bar{f}_P(q_P(t),c^*(t))}q_P^{\prime}(t),
\end{equation}
where \bee\label{optimalg2} \bar{f}_P(q_P(t),c^*(t)):=
 \left\{
 \begin{array}{lll}
 -{\lambda\over p}\overline{c}^pe^{-q_P(t)}{<0},
 & \text{if}\ q_P(t)<(p-1)\ln\overline{c}+\ln\lambda;
 \vspace{2mm}\\
 -{\lambda^{1/(1-p)}\over p}e^{q_P(t)/(p-1)}{<0},
 & \text{if}\ (p-1)\ln\overline{c}+\ln\lambda\leq q_P(t)\leq (p-1)\ln\underline{c}+\ln\lambda;
 \vspace{2mm}\\
 -{\lambda\over p}\underline{c}^pe^{-q_P(t)}{<0},
 & \text{if}\ q_P(t)>(p-1)\ln\underline{c}+\ln\lambda.
 \end{array}
 \right.
\eee

We claim that the sign of $q_P^{\prime}(t)$ does not change for
$t\in[0,T]$. Otherwise, suppose there exist $0\leq t_1<t_2\leq T$
such that $q_P^{\prime}(t_1)>0$ and $q_P^{\prime}(t_2)<0$. By the
continuity of $q_P^{\prime}(t)$, there exists $t\in(t_1,t_2)$ such
that $q_P^{\prime}(t)=0$. Now let
$t_3:=\inf\{t>t_1:q_P^{\prime}(t)=0\}$. It follows that
$t_3\in(t_1,t_2)$, $q_P^{\prime}(t_3)=0$, and $q_P^{\prime}(t)>0$ for
$t\in[t_1,t_3)$. By the Mean Value Theorem, there exits
$t_4\in(t_1,t_3)$ such that
$q_P^{\prime\prime}(t_4)=\frac{q_P^{\prime}(t_3)-q_P^{\prime}(t_1)}{t_3-t_1}<0$.
However, $q_P^{\prime\prime}(t)>0$ for $t\in[t_1,t_3)$ according to
(\ref{second_derivative}). This is a contradiction.

We have shown that $q_P(t)$ is either nonincreasing or nondecreasing
for $t\in[0,T]$. Thus, it suffices to consider the sign of
$q_P^{\prime}(T)$.

Let us first consider the case
$\underline{c}<\overline{c}<\lambda^{1/(1-p)}$. For this case, we
have $(p-1)\ln\overline{c}+\ln\lambda>0=q_P(T)$, and therefore,
{(\ref{optimalg1}) implies that} ODE (\ref{intergralequation}) at
$t=T$ reduces to
$$q_P^{\prime}(T)=-(\lambda\overline{c}^{p}-p\overline{c})-pK+\rho,$$
where the constant $K$ is given in (\ref{constantK}). However,
Theorem \ref{theorem_consumption} implies that
$c^*(t)\equiv\overline{c}$ if $\rho-pK\geq (1-p)\overline{c}$ in
this case, so we only need to consider the situation $\rho-pK<
(1-p)\overline{c}$ for the monotone property of $c^*(t)$. Together
with $\overline{c}<\lambda^{1/(1-p)}$, we further obtain that
$$q_P^{\prime}(T)<-(\overline{c}^{1-p}\overline{c}^p-p\overline{c})+(1-p)\overline{c}=0.$$
In turn, $q_P^{\prime}(t)\leq0$ for $t\in[0,T]$, which implies that
$c^*(t)$ is nondecreasing for $t\in[0,T]$.

The other two cases $\underline{c}\leq\lambda^{1/(1-p)}\leq
\overline{c}$ and $\lambda^{1/(1-p)}<\underline{c}<\overline{c}$ can
be treated in a similar way, so their proofs are omitted.
\end{proof}

\subsection{The optimal consumption under logarithm utility}

\begin{theorem}\label{theorem_consumption for log}

Assume that $T$ is a large enough number. For the logarithm utility
case, the optimal consumption $c_t^*=\widetilde{x}_{c,L}^*
(q_L(t)),\ t\in [\,0,T\,]$, is a deterministic process, with
$\widetilde{x}_{c,L}^*(\cdot)$ and $q_L(\cdot)$ given respectively
in (\ref{widehat_c for log}) and (\ref{definition of g1 and g2}).
Moreover, the optimal consumption $c^*_t$ is summarized in Table
5\footnote{Note that when $\underline{c}=0$, since $\lambda>0$ and
$\rho\geq0$, the first row and the first column about the optimal
consumption are then irrelevant.}.

\begin{table}[h]
\centering
\caption{the optimal consumption in the case of $\underline{c}\geq 0$}
\begin{tabular}{m{0.08\columnwidth} m{0.16\columnwidth}<{\centering} m{0.12\columnwidth}<{\centering} m{0.12\columnwidth}<{\centering} m{0.12\columnwidth}<{\centering} m{0.16\columnwidth}<{\centering}}
\hline \quad &{\scriptsize$0<\rho
<\underline{c}$}&{\scriptsize$\rho=\underline{c}$}
&{\scriptsize$\underline{c}<\rho<\overline{c}$}&{\scriptsize$\rho=\overline{c}$}
&{\scriptsize$\rho>\overline{c}$}\\ \hline
{\scriptsize$0<\lambda\leq \underline{c}$}
&\cellcolor{blue!25}{\scriptsize$\underline{c}$}&\cellcolor{blue!25}{\scriptsize$\underline{c}$}
&
{\scriptsize$\widehat{c}(t)I_0^2+\underline{c}I_2^3$}&{\scriptsize$\widehat{c}(t)I_0^2+\underline{c}I_2^3$}&
{\scriptsize$\overline{c}I_0^1+\widehat{c}(t)I_1^2+\underline{c}I_2^3$}\\
\hline {\scriptsize$\underline{c}<\lambda<\overline{c}$}
&{\scriptsize$\underline{c}I_0^2+\widehat{c}(t)I_2^3$}
&\cellcolor{red!25}{\scriptsize$\hat c(t)$}
&\cellcolor{red!25}{\scriptsize$\hat c(t)$}
&\cellcolor{red!25}{\scriptsize$\hat
c(t)$}&{\scriptsize$\overline{c}I_0^1+\widehat{c}(t)I_1^3$}\\ \hline
{\scriptsize$\lambda\geq\overline{c}$}
&{\scriptsize$\underline{c}I_0^2+\widehat{c}(t)I_2^1+\overline{c}I_1^3$}
&{\scriptsize$\widehat{c}(t)I_0^1+\overline{c}I_1^3$}&{\scriptsize$\widehat{c}(t)I_0^1+\overline{c}I_1^3$}&\cellcolor{yellow!25}{\scriptsize$\overline{c}$}&\cellcolor{yellow!25}{\scriptsize$\overline{c}$}\\
\hline
\end{tabular}
\end{table}
\noindent Herein, $I^a_b$ represents the indicator function of the
time interval $[\,T_a,T_b\,]$, and \be\label{T1 T2}
  \widehat{c}(t)=\lambda e^{-q_L(t)},\quad T_0=0,\quad T_{1}=T+{1\over\rho}\ln{\lambda(\rho-\overline{c})\over\overline{c}(\rho-\lambda)},\quad T_{2}=T+{1\over\rho}\ln{\lambda(\rho-\underline{c})\over\underline{c}(\rho-\lambda)},\quad T_4=T.
\ee The function $q_L$ takes the form \be\label{definition of g1 and
g2}
 q_L(t)=\ln\left[\,{\lambda\over\rho}+\left(1-{\lambda\over\rho}\right)e^{-\rho(T-t)}\,\right].
\ee

Moreover, the optimal consumption $c^*_t$ is nonincreasing with
respect to $t$ for $\rho\geq\lambda$, and nondecreasing with respect
to $t$ for $\rho\leq\lambda$.
\end{theorem}

\begin{proof} First, it is clear that the solution of ODE (\ref{intergralequation1 for log}) takes the form (\ref{definition of g1 and g2}).
From (\ref{widehat_c for log}), we know that
$\widetilde{x}^*_{c,L}(x_q)$ is nonincreasing with respect to $x_q$.
Moreover, the expression (\ref{definition of g1 and g2}) implies
that $q_L(\cdot)$ is nondecreasing with respect to $t$ when
$\rho\geq\lambda$, and nonincreasing with respect to $t$ when
$\rho\leq\lambda$. Then, the monotonicity of
$c^*_t=\widetilde{x}^*_{c,L}(q_L(t))$ follows immediately.

Next, we note that \be\label{intial terminal value}
 e^{q_L(T)}=1\left\{
 \begin{array}{ll}
 \geq{\lambda\over\underline{c}},&0<\lambda\leq \underline{c};
 \\[2mm]
 \in({\lambda\over\overline{c}},{\lambda\over\underline{c}}),&\underline{c}<\lambda<\overline{c};
 \\[2mm]
 \leq{\lambda\over\overline{c}},&\lambda\geq\overline{c},
 \end{array}
 \right. \quad  \lim\limits_{T\rightarrow+\infty}e^{q_L(0)}
 ={\lambda\over\rho}\left\{
 \begin{array}{ll}
 \geq{\lambda\over\underline{c}},&0<\rho\leq \underline{c};
 \\[2mm]
 \in({\lambda\over\overline{c}},{\lambda\over\underline{c}}),&\underline{c}<\rho<\overline{c};
 \\[2mm]
 \leq{\lambda\over\overline{c}},&\rho\geq\overline{c}.
 \end{array}
 \right.
\ee In the following, we only prove the case
$0<\lambda\leq\underline{c}$ and $\rho>\overline{c}$. Other cases
follow along similar arguments. It follows from \eqref{intial
terminal value} that
$$
  \widehat{c}(T)=\lambda e^{-q_L(T)}\leq \underline{c}<\overline{c}<\lambda e^{-q_L(0)}=\widehat{c}(0),\quad
  c^*_T=\widetilde{x}^*_{c,L}(q_L(T))=\underline{c},\quad c^*_0=\widetilde{x}^*_{c,L}(q_L(0))=\overline{c}
$$
provided $T$ is large enough. Moreover, since $q_L(\cdot)$ is continuous and strictly increasing with respect to $t$, there exists unique $(T_1,T_2)$ such that $$
  q_L(t)\geq \ln{\lambda\over\underline{c}},\;\;t\in[\,T_2,T\,];\qquad
  \ln{\lambda\over\overline{c}}<q_L(t)<\ln{\lambda\over\underline{c}},\;\;t\in(T_1,T_2);\qquad
  q_L(t)\leq \ln{\lambda\over\overline{c}},\;\;t\in[\,0,T_1\,],
$$
and $T_1,T_2$ take the form in~\eqref{T1 T2}. Together with
\eqref{widehat_c for log}, we deduce that
$c_t^*=\widetilde{x}^*_{c,L}(q_L(t))
=\overline{c}I_0^1+\widehat{c}(t)I_1^2+\underline{c}I_2^3$.
\end{proof}

\section{The impacts of model uncertainty, portfolio-consumption constraints and borrowing costs}

In this section, we investigate the impacts of model uncertainty,
portfolio-consumption constraints and borrowing costs on the
worst-case parameters $(\mu^*,\sigma^*)$ and the optimal
portfolio-consumption strategies $(\pi^*,c^*)$.

\begin{proposition}\label{proposition2}
Suppose that Assumption 4.1 holds. Then, for the power utility case,
the worst-case parameters and the optimal portfolio-consumption
strategies admit the following monotone properties in terms of the
borrowing rate $R$, the constraint set
$[\underline{\pi},\overline{\pi}]\times[\underline{c},\overline{c}]$,
and the uncertain parameter set
$[\underline{\mu},\overline{\mu}]\times[\underline{\sigma}^2,\overline{\sigma}^2]$,
as specified in Table 6. The symbols $\searrow$, $\nearrow$, $\perp$
and NM represent nonincreasing, nondecreasing, independent and
non-monotone of the corresponding  variable. For example, the bottom
row and the first column (left-bottom corner) means $c_s^*$ is
nondecreasing in the borrowing rate  $R$.

\begin{table}[h!]
    \caption{the comparative statistics}\hspace{0.5cm}
    \hspace{1cm}
    \ \ \ \ \ \ \ \ \ \ \ \
    \begin{tabular}{c c c c c c c c c c}
       \hline
       &$R$ & $\underline{\pi}$ &$\overline{\pi}$ & $\underline{c}$ & $\overline{c}$ & $\underline{\mu}$ &$\overline{\mu}$ & $\underline{\sigma}$ & $\overline{\sigma}$  \\
       \hline
      $\mu^*_s$ & $\perp$ &  $\perp$ &  $\perp$ &  $\perp$  &  $\perp$  & $\nearrow$ & $\nearrow$ & $\perp$ & $\perp$  \\
      $\sigma^*_s$ & $\perp$ & $\perp$ & $\perp$ & $\perp$ & $\perp$  & $\perp$ & $\perp$ & $\perp$ & $\nearrow$\\
      $\pi^*_s$ & $\searrow$ &  $\nearrow$ &  $\nearrow$ &  $\perp$  &  $\perp$  & $\nearrow$ & $\nearrow$ & $\perp$ & $\searrow$  \\
      $c^*_s$ & $\nearrow$ & $\nearrow$ & $\searrow$ & $\nearrow$ & NM  & $\searrow$ & $\nearrow$ & $\perp$ & $\nearrow$\\
       \hline
    \end{tabular}  \label{table4}
  \end{table}

\end{proposition}

Before proceeding to the proof, we provide some intuitive
explanations for the above results. The impacts of different
parameters on the worst-case parameters $(\mu_s^*,\sigma^*_s)$ and
the optimal portfolio $\pi_s^*$ are obvious from {the results} in
Theorem \ref{theorem_portfolio}. So we only discuss about their
impacts on the optimal consumption $c^*_s$.

By the expression (\ref{widehat_c}) and
$c^*_s=\widetilde{x}_{c,P}^*(q_P(s))$, the parameters ($R$,
$\underline{\pi}$, $\overline{\pi}$, $\underline{\mu}$,
$\overline{\mu}$, $\underline{\sigma}$, $\overline{\sigma}$) will
effect the optimal consumption through the channel of the
opportunity process $e^{q_P(s)}$, which is the investor's extra
utilities obtained by optimizing over all the admissible
portfolio-consumption strategies (least affected by model
uncertainty) in the remaining horizon $[s,T]$. A closer look at the
ODE (\ref{intergralequation}) for $q_P(s)$ tells us that those
parameters will only enter into the future investment contributing
factor $g(x_{\pi}^*;x_{\mu}^*,{x_\sigma^*})$ in (\ref{definition of
g}). Increasing the borrowing cost $R$ will make the future
investment contributing factor $g(x_{\pi}^*;x_{\mu}^*,{x_\sigma^*})$
smaller, so the opportunity process will also become smaller, i.e.
the investor will obtain less utilities in the remaining horizon. In
turn, her current optimal consumption will go up. Similarly,
enlarging the uncertainty parameters interval
$[\underline{\mu},\overline{\mu}]\times[\underline{\sigma}^2,\overline{\sigma}^2]$
or shrinking the portfolio constraint interval
$[\underline{\pi},\overline{\pi}]$ will also make the future
investment contributing factor $g(x_{\pi}^*;x_{\mu}^*,{x_\sigma^*})$
smaller, and therefore, the current optimal consumption will arise.

The more striking result is probably the impact of the consumption
constraint interval $[\underline{c},\overline{c}]$ on the optimal
consumption $c^*_s$. Note that the constraint interval will only
effect the consumption contributing factor $f_P(q_P(s),c^*(s))$ in
(\ref{definition of f}), with
{$c^*(s)=\widetilde{x}_{c,P}^*(q_P(s))$}. The smaller interval will
lead to a smaller consumption contributing factor
$f_P(q_P(s),c^*(s))$ as indicated in {(\ref{optimalconsumption})}.
In turn, the investor will obtain less utilities in the remaining
horizon $[s,T]$. This might suggest the current optimal consumption
would increase. However, it is not always the case, as there is less
chance for the unconstrained optimal consumption $\widehat{c}(s)$ to
stay in the shrinking interval $[\underline{c},\overline{c}]$. If
$\widehat{c}(s)$ reaches the lower bound $\underline{c}$, then the
optimal consumption will further arise as $\underline{c}$ increases.
On the other hand, if $\widehat{c}(s)$ reaches the upper bound
$\overline{c}$, then the optimal consumption will go down for
$\overline{c}$ becomes smaller, thus offsets the previous increasing
impact on the optimal consumption when $\overline{c}$ is decreasing.
\emph{This means the optimal consumption is non-monotone in its
upper bound $\overline{c}$.}

\begin{proof}

(i) \emph{The monotone property of $\mu_s^*$}. According to Theorem
\ref{theorem_portfolio}, the worst-case drift can be rewritten as
\bee {\mu_s^*}&=&\underline{\mu}1_{\{\overline{\mu}\geq r\}}+
\overline{\mu}1_{\{\overline{\mu}<r\}}\\
&=&\underline{\mu}1_{\{\underline{\mu}>r\}}+
\overline{\mu}1_{\{\underline{\mu}\leq r\}}\eee {for} $s\in[0,T]$.
The first line implies that $\mu_s^*$ is nondecreasing in
$\underline{\mu}$, and the second line implies it is also
nondecreasing in $\overline{\mu}$, and is irrelevant to the other
parameters
$(R,\underline{\pi},\overline{\pi},\underline{c},\overline{c},\underline{\sigma},\overline{\sigma})$.

(ii) \emph{The monotone property of $\sigma_s^*$}. The conclusion
simply follows from the expression of the worst-case volatility
$\sigma_s^*=\overline{\sigma}$ for $s\in[0,T]$ in Theorem
\ref{theorem_portfolio}.

(iii) \emph{The monotone property of $\pi_s^*$}. First, the
expressions of $\beta^1,\beta^2,\beta^3$ in Theorem
\ref{theorem_portfolio} imply that they are all nondecreasing in
$\underline{\mu},\overline{\mu}$, nonincreasing in
$R,\overline{\sigma}$ and independent of
$\underline{\sigma},\underline{c},\overline{c}$, so is the optimal
portfolio $\pi_s^*$, {as $\pi_s^*$ is nondecreasing with respect to
$\beta^1,\beta^2,\beta^3$ (cf. Figure \ref{figureinvestment})}.

From Table 1 in Theorem \ref{theorem_portfolio}, we further obtain
\bee \pi_s^*&=&\min\{\beta^1,\overline{\pi}\}1_{\{\beta^1\geq1\}}+
C_11_{\{\beta^1<1\}}\\
&=&\max\{\beta^3,\underline{\pi}\}1_{\{\beta^3\leq 0\}}+
C_21_{\{\beta^3>0\}} \eee for some constants $C_1$ independent of
$\overline{\pi}$, and $C_2$ independent of $\underline{\pi}$. Hence,
$\pi_s^*$ is nondecreasing in both $\underline{\pi}$ and
$\overline{\pi}$.

(iv) \emph{The monotone property of $c_s^*$}. We first study the
impacts of different parameters on the solution $q_P(t)$ of ODE
(\ref{intergralequation}). Note that ($R$, $\underline{\pi}$,
$\overline{\pi}$, $\underline{\mu}$, $\overline{\mu}$,
$\underline{\sigma}$, $\overline{\sigma}$) will effect $q_P(t)$ only
through $g(x_{\pi}^*;x_{\mu}^*,{x_\sigma^*})=K$, where $K$ is given
in (\ref{constantK}).

It is obvious from the expression (\ref{constantK}) that $K$ is
nonincreasing in $R$. Moreover, since $K$ is the maximum value of
$g(x_{\pi};x_{\mu}^*,{x_\sigma^*})$ over
$x_{\pi}\in[\underline{\pi},\overline{\pi}]$, $K$ is nonincreasing
in $\underline{\pi}$ and nondecreasing in $\overline{\pi}$. On the
other hand, $K$ is also the minimum value of
$g(x_{\pi}^*;x_{\mu},{x_\sigma})$ over
$(x_{\mu},{x_\sigma})\in[\underline{\mu},\overline{\mu}]\times[\underline{\sigma}^2,\overline{\sigma}^2]$.
Therefore, $K$ is nondecreasing in
$\underline{\mu},\underline{\sigma}$ and nonincreasing in
$\overline{\mu},\overline{\sigma}$. However, the expression of $K$
further implies that $K$ is independent of $\underline{\sigma}$.

It then follows from the comparison theorem for ODE {
(\ref{intergralequation})} that its solution $q_P(s)$ is nonincreasing
in $R$, $\underline{\pi}$, $\overline{\mu}$, $\overline{\sigma}$,
nondecreasing in $\overline{\pi}$ and $\underline{\mu}$, and
independent of $\underline{\sigma}$. The conclusion about the
optimal consumption $c^*_s$ then follows from Theorem \ref{Theorem1 for power}
together with the expression {(\ref{widehat_c}).}

In terms of the impacts of $\underline{c}$ and $\overline{c}$ on
$c_s^*$, since $f_P(q_P(s),c^*(s))$, with
{$c^*(s)=\widetilde{x}_{c,P}^*(q_P(s))$}, is the maximum value of
$f_P(q_P(s),x_c)$ over $x_c\in[\underline{c},\overline{c}]$, it is
nonincreasing in $\underline{c}$ and nondecreasing in
$\overline{c}$. Following the comparison theorem for ODE
(\ref{intergralequation}) and the expression (\ref{widehat_c}) once
again, we conclude $\widehat{c}(s)$ is nondecreasing in
$\underline{c}$ and nonincreasing in $\overline{c}$.

In turn, the expression {(\ref{widehat_c})} implies that the optimal
consumption $c^*_s$ is also nondecreasing in $\underline{c}$, but
neither increasing nor decreasing in $\overline{c}$, for the second
and last terms in $c^*(s)$ offset the effects of each other. Indeed,
{we show the non-monotonicity in the case of $0\leq
\underline{c}<\overline{c}_2<\overline{c}_1<\lambda^{1/(1-p)}$ and
$\rho-pK\in ((1-p)\underline{c},(1-p)\overline{c}_2)$. According to
Theorem \ref{theorem_consumption}, both $c_1^*(t)$ and $c_2^*(t)$
take the form of $\widehat{c}(t)I^{12}_0+\overline{c}I^4_{12}$. When
$t$ is close to $T$, then
$c_1^*(t)=\overline{c}_1>\overline{c}_2=c_2^*(t)$. On the other
hand, when $T$ is large enough and $t$ is close to zero, we have
$$
 c^*_1(t)=\widehat{c}_1(t)=\exp\left\{{q_{P,1}(t)\over p-1}\right\}<\exp\left\{{q_{P,2}(t)\over
 p-1}\right\}=\widehat{c}_2(t)=c^*_2(t),
$$
where the strict inequality can be derived from the comparison
theorem for ODE.}
\end{proof}

Finally, we present the result for the logarithm utility case. Its
proof is omitted as it is similar to the proof for the power utility
case.

\begin{proposition}\label{proposition3}
Suppose that Assumption 4.1 holds. Then, for the logarithm utility
case, the worst-case parameters and the optimal
portfolio-consumption strategies have the following monotone
properties in terms of the borrowing rate $R$, the constraint set
$[\underline{\pi},\overline{\pi}]\times
[\underline{c},\overline{c}]$, and the uncertain parameter set
$[\underline{\mu},\overline{\mu}]\times
[\underline{\sigma}^2,\overline{\sigma}^2]$, as showed in Table 7.
The symbols $\searrow$, $\nearrow$, $\perp$\ represent
nonincreasing, nondecreasing and independent of the corresponding
variable.

\begin{table}[h]
\centering
\caption{The comparative statistics for log}
\begin{tabular}{lccccccccc}
\hline
\quad &$R$ &$\underline{\pi}$ &$\overline{\pi}$ &$\underline{c}$ &$\overline{c}$ &$\underline{\mu}$ &$\overline{\mu}$ &$\underline{\sigma}$  &$\overline{\sigma}$\\
\hline
$\mu^*$ &$\perp$ &$\perp$ &$\perp$ &$\perp$ &$\perp$ &$\nearrow$ &$\nearrow$ &$\perp$ &$\perp$\\
$\sigma^*$ &$\perp$ &$\perp$ &$\perp$ &$\perp$ &$\perp$ &$\perp$ &$\perp$ &$\perp$ &$\nearrow$ \\
$\pi^*$ &$\searrow$ &$\nearrow$ &$\nearrow$ &$\perp$ &$\perp$ &$\nearrow$ &$\nearrow$ &$\perp$ &$\searrow$\\
$c^*$ &$\perp$ &$\perp$ &$\perp$ &$\nearrow$ &$\nearrow$ &$\perp$ &$\perp$ &$\perp$ &$\perp$\\
\hline
\end{tabular}
\end{table}
\end{proposition}

\appendix
\section{Appendix: Proof of Theorem \ref{theorem_portfolio}}

\noindent{\it Proof of Theorem \ref{theorem_portfolio}.} According
to Theorem \ref{Theorem1 for power}, if
$\{x_{\pi}^*;x_{\mu}^*,{x_\sigma^*}\}$ is a saddle point of the
function $g(\cdot;\cdot,\cdot)$, then $x_{\pi}^*$ is the optimal
portfolio, and $(x_{\mu}^*,{\sqrt{x_{\sigma}^*}})$ are the
worst-case parameters. {Thus, it is sufficient to show that
$\{\pi^*;\mu^*,{x_\sigma^*}\}$ given in Theorem
\ref{theorem_portfolio} is indeed a saddle point of the function
$g(\cdot;\cdot,\cdot)$.}

First, for fixed $x_{\pi}\in[\underline{\pi},\overline{\pi}]$, it is
obvious to check that with
$$
  g(x_{\pi};x_{\mu},{x_\sigma})={p-1\over 2}\,{x_\sigma}x_{\pi}^2
 +\left(x_{\mu}x_{\pi}+r(1-x_{\pi})-(R-r)(1-x_{\pi})^-\right),
$$
we have \be\label{min}
\min_{(x_{\mu},{x_\sigma})\in[\underline{\mu},\overline{\mu}]\times[\underline{\sigma}^2,\overline{\sigma}^2]}g(x_{\pi};x_{\mu},{x_\sigma})=
 \left\{
 \begin{array}{lll}
g(x_{\pi};{\underline{\mu}},{\overline{\sigma}}^2),
 &\ \text{if}\ x_{\pi}>0;
 \vspace{2mm}\\
 g(x_{\pi};[\underline{\mu},\overline{\mu}],{\overline{\sigma}}^2),
 &\ \text{if}\ x_{\pi}=0;
 \vspace{2mm}\\
 g(x_{\pi};{\overline{\mu}},{\overline{\sigma}}^2),
 &\ \text{if}\ x_{\pi}<0,
 \end{array}
 \right.
\ee where $[\underline{\mu},\overline{\mu}]$ means that $x_{\mu}^*$
may take any value in that interval.

The above minimum function can be further written in a compact form
by defining \bee
  g_1(x_{\pi})&:=& g(x_{\pi};x_{\mu}^*,{x_\sigma^*})
  \\[2mm]
  &=&{p-1\over 2}\,\overline{\sigma}^2x_{\pi}^2
  +\Big(\,\underline{\mu}I_{\{{x_{\pi}}>0\}}+\overline{\mu}I_{\{x_{\pi}<0\}}-rI_{\{x_{\pi}<1\}}-RI_{\{x_{\pi}\geq
  1\}}\,\Big)x_{\pi}
  \\[2mm]
  &&+\,\Big(\,rI_{\{x_{\pi}<1\}}+RI_{\{x_{\pi}\geq 1\}}\,\Big).
\eee In the following, we study the maximum value of $g_1(x_{\pi})$
in three different cases $x_{\pi}\geq1$, $0\leq x_{\pi}\leq1$ and
$x_{\pi}\leq0$, then together with the constraint
$\underline{\pi}\,\leq x_{\pi} \leq \overline{\pi}$, we will obtain
the maximizer $x_{\pi}^*$ and the associated maximum value
$g_1(x_{\pi}^*)$.

 \noindent \underline{Case (1)} $\,x_{\pi}\geq 1$.
$$
 g_1(x_{\pi})={p-1\over2}\,\overline{\sigma}^2\,\left[\,x_{\pi}
  +{\underline{\mu}-R\over (p-1)\overline{\sigma}^2}\,\right]^2+R-{(\underline{\mu}-R)^2\over 2(p-1)\overline{\sigma}^2}.
$$
If $\beta_1=(\underline{\mu}-R)/((1-p)\overline{\sigma}^2)\geq
\overline{\pi}$, then
$$
 \max\limits_{1\leq x_{\pi}\leq \overline{\pi}}\,g_1(x_{\pi})=g_1(\overline{\pi})
\geq g_1(1).
$$
If $1<\beta_1<\overline{\pi}$, then
$$
 \max\limits_{1\leq x_{\pi}\leq \overline{\pi}}\,g_1(x_{\pi})=g_1(\beta_1)
 >g_1(1).
$$
If $\beta_1\leq1$, then
$$
 \max\limits_{1\leq x_{\pi}\leq \overline{\pi}}\,g_1(x_\pi)=g_1(1).
$$

\noindent \underline{Case (2)} $\,0\leq x_{\pi}\leq 1$.
$$
 g_1(x_{\pi})={p-1\over2}\,\overline{\sigma}^2\,\left[\,x_{\pi}
  +{\underline{\mu}-r\over (p-1)\overline{\sigma}^2}\,\right]^2+r-{(\underline{\mu}-r)^2\over 2(p-1)\overline{\sigma}^2}.
$$
If $\beta_2=(\underline{\mu}-r)/((1-p)\overline{\sigma}^2)\geq 1$,
then
$$
 \max\limits_{0\leq x_{\pi}\leq 1}\,g_1(x_\pi)=g_1(1)>g_1(0).
$$
If $0<\beta_2<1$, then
$$
 \max\limits_{0\leq x_{\pi}\leq 1}\,g_1(x_{\pi})=g_1(\beta_2)
 >\max\{g_1(1),g_1(0)\}.
$$
If $\beta_2\leq 0$, then
$$
 \max\limits_{0\leq x_{\pi}\leq 1}\,g_1(x_{\pi})=g_1(0)>g_1(1).
$$

\noindent \underline{Case (3)} $\,x_{\pi}\leq 0$.
$$
 g_1(x_{\pi})={p-1\over2}\,\overline{\sigma}^2\,\left[\,x_{\pi}
 +{\overline{\mu}-r\over (p-1)\overline{\sigma}^2}\,\right]^2+r-{(\overline{\mu}-r)^2\over 2(p-1)\overline{\sigma}^2}.
$$
If $\beta_3=(\overline{\mu}-r)/((1-p)\overline{\sigma}^2)\geq 0$,
then
$$
 \max\limits_{\underline{\pi}\,\leq x_{\pi}\leq 0}\,g_1(x_{\pi})=g_1(0).
$$
If $\underline{\pi}<\beta_3<0$, then
$$
 \max\limits_{\underline{\pi}\,\leq x_{\pi}\leq 0}\,g_1(x_{\pi})=g_1(\beta_3)
 >g_1(0).
$$
If $\beta_3\leq \underline{\pi}$, then
$$
 \max\limits_{\underline{\pi}\,\leq x_{\pi}\leq 0}\,g_1(x_{\pi})=g_1(\underline{\pi})
 \geq g_1(0).
$$

Comparing the maximum values in the above three cases, and noting
that the fact $\beta_1\leq\beta_2\leq\beta_3$, we see that
$\max\limits_{\underline{\pi}\,\leq x_{\pi}\leq
\overline{\pi}}\,g_1(x_{\pi})=g_1(x_{\pi}^*)=K$, where $K$ {is
defined in} \eqref{constantK}, and the optimal $x_{\pi}^*$ {is
defined in} Table \ref{table1}. Thus, we have proved
$$
  {g}(x^*_\pi;x_\mu^*,{x_\sigma^*})\geq
 {g}(x_\pi;x^*_\mu, {x_\sigma^*}),\;\;\forall\;x_\pi\in[\,\underline{\pi},\overline{\pi}\,].
$$

On the other hand, with $x_{\pi}^*$ as in Table 1, it follows from
(\ref{min}) that
$$
 {g}(x^*_\pi;x_\mu^*,{x_\sigma^*})\leq
  {g}(x^*_\pi;x_\mu,{x_\sigma}),\;\;\forall\;(x_{\mu},{x_\sigma})\in[\,\underline{\mu},\overline{\mu}\,]\times[\,\underline{\sigma}^2,\overline{\sigma}^2\,].
$$
Hence, $\{x_{\pi}^*;x_{\mu}^*,{x_\sigma^*}\}$ is a saddle point of
the function $g(\cdot;\cdot,\cdot)$. $\hfill{} \Box$
\medskip

\section{Appendix: Proof of Theorem \ref{theorem_consumption}}

First, we give the explicit solution to ODE
(\ref{intergralequation}) in Table 8 when
$\underline{c}>0$\footnote{In the case of $\underline{c}=0$, the
results are similar to those in Table \ref{table5} except that
$q_{123}$ and $q_{23}$ are replaced by $q_{12}$ and $q_{2}$,
respectively. Note that in this case, the forth and fifth rows in
Table \ref{table5} and $q_{123}, q_{23}, q_{32}, q_{321}, q^3,
T_{123}, T_{23}, T_{32}, T_{321}$ are irrelevant.}.

\begin{table}[h!]
  \centering
    \caption{The explicit solution {$q_P(\cdot)$} to ODE (\ref{intergralequation}) in the case of $\underline{c}>0$
}
    \begin{tabular}{c c c c c c }
       \hline
       & {\scriptsize$\rho-pK$} & & & &   \\
       &{\scriptsize$(-\infty,(1-p)\underline{c}\,)$} & {\scriptsize$\{(1-p)\underline{c}\,\}$}
       &{\scriptsize$((1-p)\underline{c},(1-p)\overline{c}\,)$} & {\scriptsize$\{(1-p)\overline{c}\,\}$} &
       {\scriptsize$((1-p)\overline{c},+\infty)$}   \\
       \hline
       {\scriptsize$\underline{c}<\overline{c}<\lambda^{1/(1-p)}$} & {\scriptsize$ q_{123}(t)$} & {\scriptsize$ q_{12}(t)$} &
       {\scriptsize$ q_{12}(t)$} & {\scriptsize\cellcolor{yellow!25}$ q_{1}(t)$} & {\scriptsize\cellcolor{yellow!25}$  q_{1}(t)$}  \\
       {\scriptsize$\underline{c}<\overline{c}=\lambda^{1/(1-p)}$} & {\scriptsize$ q_{23}(t)$} & \cellcolor{red!25}{\scriptsize$ q_2(t)$} &
       \cellcolor{red!25}{\scriptsize$ q_2(t)$} & \cellcolor{yellow!25}{\scriptsize$ q_1(t)$} & \cellcolor{yellow!25}
       {\scriptsize$ q_1(t)$}  \\
       {\scriptsize$ \underline{c}<\lambda^{1/(1-p)}<\overline{c}$} & {\scriptsize$ q_{23}(t)$} & \cellcolor{red!25}{\scriptsize$ q_2(t)$} &
       \cellcolor{red!25}{\scriptsize$ q_2(t)$} & \cellcolor{red!25}{\scriptsize$ q_2(t)$} & {\scriptsize$ q_{21}(t)$}  \\
       {\scriptsize$ \lambda^{1/(1-p)}=\underline{c}<\overline{c}$} & \cellcolor{blue!25}{\scriptsize$ q_3(t)$} & \cellcolor{blue!25}{\scriptsize$ q_3(t)$}
       & \cellcolor{red!25}{\scriptsize$ q_2(t)$} & \cellcolor{red!25}{\scriptsize$ q_2(t)$} &{\scriptsize$ q_{21}(t)$}  \\
       {\scriptsize$ \lambda^{1/(1-p)}<\underline{c}<\overline{c}$} & \cellcolor{blue!25}{\scriptsize$ q_3(t)$} & \cellcolor{blue!25}{\scriptsize$ q_3(t)$} &
       {\scriptsize$ q_{32}(t)$} & {\scriptsize$ q_{32}(t)$} & {\scriptsize$ q_{321}(t)$}  \\
\hline
    \end{tabular}  \label{table5}
\end{table}

\noindent The solutions
$q_{123},q_{12},q_1,q_{23},q_{2},q_{21},q_3,q_{32},q_{321}$ have the
explicit forms
 \bee
  q_{123}(t)&\!\!\!\!\!\!=\!\!\!\!\!\!&q^1(t;1/\lambda,T_{12},T)I_{[\,T_{12},T\,]}
  +q^2(t;\overline{c}^{p-1},T_{123},T_{12})I_{[\,T_{123},T_{12}\,]}
  +q^3(t;\underline{c}^{p-1},0,T_{123})I_{[\,0,T_{123}\,]};
  \\[2mm]
  q_{12}(t)&\!\!\!\!\!=\!\!\!\!\!&q^1(t;1/\lambda,T_{12},T)I_{[\,T_{12},T\,]}
  +q^2(t;\overline{c}^{p-1},0,T_{12})I_{[\,0,T_{12}\,]};
  \\[2mm]
  q_{1}(t)&\!\!\!\!\!=\!\!\!\!\!&q^1(t;1/\lambda,0,T);\qquad
  q_2(t)=q^2(t;1/\lambda,0,T);\qquad
  q_3(t)=q^3(t;1/\lambda,0,T);
  \\[2mm]
  q_{23}(t)&\!\!\!\!\!=\!\!\!\!\!&q^2(t;1/\lambda,T_{23},T)I_{[\,T_{23},T\,]}+q^3(t;\underline{c}^{p-1},0,T_{23})I_{[\,0,T_{23}\,]};
  \\[2mm]
  q_{21}(t)&\!\!\!\!\!=\!\!\!\!\!&q^2(t;1/\lambda,T_{21},T)I_{[\,T_{21},T\,]}+q^1(t;\overline{c}^{p-1},0,T_{21})I_{[\,0,T_{21}\,]};
  \\[2mm]
  q_{32}(t)&\!\!\!\!\!=\!\!\!\!\!&q^3(t;1/\lambda,T_{32},T)I_{[\,T_{32},T\,]}+q^2(t;\underline{c}^{p-1},0,T_{32})I_{[\,0,T_{32}\,]};
  \\[2mm]
  q_{321}(t)&\!\!\!\!\!=\!\!\!\!\!&q^3(t;1/\lambda,T_{32},T)I_{[\,T_{32},T\,]}+q^2(t;\underline{c}^{p-1},T_{321},T_{32})I_{[\,T_{321},T_{32}\,]}
  +q^1(t;\overline{c}^{p-1},0,T_{321})I_{[\,0,T_{321}\,]},
\eee where $I_{[\underline{T},\overline{T}]}$ is an indicator
function of the set $[\underline{T},\overline{T}]$, and the
functions
$q^1(t;A,\underline{T},\overline{T}),q^2(t;A,\underline{T},\overline{T}),$\\$q^3(t;A,\underline{T},\overline{T})$
in the interval $[\,\underline{T},\overline{T}\,]$ are given as \be
 q^1(t;A,\underline{T},\overline{T})&\!\!\!=\!\!\!&\ln\lambda+
 \left\{
 \begin{array}{ll}
 \ln\left[\,\left(\,A-{\overline{c}^p\over \rho+p\overline{c}-pK}\,\right)
 e^{(\rho+p\overline{c}-pK)(t-\overline{T})}+{\overline{c}^p\over \rho+p\overline{c}-pK}\,\right],\!
 &\!\!\rho-pK\neq -p\overline{c};
 \vspace{2mm} \\
 \ln\Big[\,A+\overline{c}^p(\overline{T}-t)\,\Big],\!
 &\!\!\rho-pK=-p\overline{c};
 \end{array}
 \right.\label{equ20}
\\[2mm]
 q^2(t;A,\underline{T},\overline{T})&\!\!\!\!=\!\!\!\!&\ln\lambda+
 \left\{
 \begin{array}{ll}
 (1-p)\ln\left[\,\left(\,A^{1/(1-p)}-{1-p\over\rho-pK}\,\right)
 e^{{\rho-pK\over 1-p}\left(\,t-\overline{T}\,\right)}+{1-p\over\rho-pK}\,\right],\!
 &\!\!\rho-pK\neq0;
 \vspace{2mm} \\
 (1-p)\ln\left[\,A^{1/(1-p)}+\overline{T}-t\,\right],\!
 &\!\!\rho-pK=0;
 \end{array}
 \right.\label{equ21}
\\[2mm]
 q^3(t;A,\underline{T},\overline{T})&\!\!\!\!=\!\!\!\!&\ln\lambda+
 \left\{
 \begin{array}{ll}
 \ln\left[\,\left(\,A-{\underline{c}^p\over \rho+p\underline{c}-pK}\,\right)
 e^{(\rho+p\underline{c}-pK)(t-\overline{T})}+{\underline{c}^p\over \rho+p\underline{c}-pK}\,\right],\!
 &\!\!\rho-pK\neq-p\underline{c};
 \vspace{2mm} \\
 \ln\left[\,A+\underline{c}^p(\overline{T}-t)\,\right],\!
 &\!\!\rho-pK=-p\underline{c},
 \end{array}
 \right.\label{equ22}
\ee and $T_{12},T_{123},T_{23},T_{21},T_{32},T_{321}$ are given as
\be
 T_{12}&=&
 \left\{
 \begin{array}{ll}
 T+{1\over\rho+p\overline{c}-pK}\left[\,
 \ln\left|\,\overline{c}^{p-1}-{\overline{c}^p\over \rho+p\overline{c}-pK}\,\right|
 -\ln\left|\,{1\over\lambda}-{\overline{c}^p\over \rho+p\overline{c}-pK}\,\right|\,\right],\;
 \;&\rho-pK\neq -p\overline{c};
 \vspace{2mm} \\
 T-1/\overline{c}+1/(\lambda\overline{c}^p),\;&\rho-pK=-p\overline{c};
 \end{array}
 \right.\label{equ23}
 \\[2mm]
 T_{123}&=&
 \left\{
 \begin{array}{ll}
 T_{12}+{1-p\over\rho-pK}\left[\,
 \ln\left|\,{1\over\underline{c}}-{1-p\over\rho-pK}\,\right|
 -\ln\left|\,{1\over\overline{c}}-{1-p\over\rho-pK}\,\right|\,\right],\quad &\rho-pK\neq0;
 \vspace{2mm} \\
 T_{12}+1/\overline{c}-1/\underline{c}\,, \;&\rho-pK=0;
 \end{array}
 \right.\label{equ24}
 \\[2mm]
 T_{23}&=&
 \left\{
 \begin{array}{ll}
 T+{1-p\over\rho-pK}\left[\,
 \ln\left|\,{1\over\underline{c}}-{1-p\over\rho-pK}\,\right|
 -\ln\left|\,\lambda^{1/(p-1)}-{1-p\over\rho-pK}\,\right|\,\right],\quad &\rho-pK\neq0;
 \vspace{2mm} \\
 T+\lambda^{1/(p-1)}-1/\underline{c}\,,
 \;&\rho-pK=0;
 \end{array}
 \right.\label{equ25}
 \\[2mm]
 T_{21}&=&T+{1-p\over\rho-pK}\left[\,
 \ln\left(\,{1\over\overline{c}}-{1-p\over\rho-pK}\,\right)
 -\ln\left(\,\lambda^{1/(p-1)}-{1-p\over\rho-pK}\,\right)\,\right];\label{equ26}
 \\[2mm]
 T_{32}&=&T+{1\over\rho+p\underline{c}-pK}\left[\,
 \ln\left(\,\underline{c}^{p-1}-{\underline{c}^p\over \rho+p\underline{c}-pK}\,\right)
 -\ln\left(\,{1\over\lambda}-{\underline{c}^p\over
 \rho+p\underline{c}-pK}\,\right)\,\right];\label{equ27}
 \\[2mm]
 T_{321}&=&T_{32}+{1-p\over\rho-pK}\left[\,
 \ln\left(\,{1\over\overline{c}}-{1-p\over\rho-pK}\,\right)
 -\ln\left(\,{1\over\underline{c}}-{1-p\over\rho-pK}\,\right)\,\right].\label{equ28}
\ee

It is routine to check that for any $A>0$ and $0\leq
\underline{T}\leq\overline{T}$, the functions
$q^1(t;A,\underline{T},\overline{T})$,
$q^2(t;A,\underline{T},\overline{T})$ and
$q^3(t;A,\underline{T},\overline{T})$ solve the following ODEs,
respectively,
 \be\label{ODE1}
  {q_P}(t)&=&\ln\lambda+\ln A+\int_t^{\overline{T}}
  \Big[\,-\rho+\lambda\overline{c}^pe^{-{q_P}(s)}-p\overline{c}+pK\,\Big]ds,\;\forall\;t\in[\,\underline{T},\overline{T}\,];
  \\[2mm]\label{ODE2}
  {q_P}(t)&=&\ln\lambda+\ln A+\int_t^{\overline{T}}
  \Bigg[\,-\rho+(1-p)\lambda^{1\over(1-p)}\exp\left\{{{q_P}(s)\over p-1}\right\}+pK\,\Bigg]ds,\;\forall\;t\in[\,\underline{T},\overline{T}\,];
  \\[2mm]\label{ODE3}
  {q_P}(t)&=&\ln\lambda+\ln A+\int_t^{\overline{T}}
  \Big[\,-\rho+\lambda\underline{c}^pe^{-{q_P}(s)}-p\underline{c}+pK\,\Big]ds,\;\forall\;t\in[\,\underline{T},\overline{T}\,].
 \ee

When $\underline{c}>0$,  $q^1(0;A,0,\overline{T}),\,
q^2(0;A,0,\overline{T})$ and $q^3(0;A,0,\overline{T})$ have the
following asymptotic properties,
 \be\nonumber
  &&\lim\limits_{\overline{T}\rightarrow\infty}q^1(0;A,0,\overline{T})=
  \left\{
  \begin{array}{ll}
  \ln\lambda+\ln\left(\,{\overline{c}^p\over \rho+p\overline{c}-pK}\,\right),
  &\rho-pK>-p\overline{c};
  \vspace{2mm} \\
  +\infty,
  &\rho-pK\leq-p\overline{c};
  \end{array}
  \right.
  \\[2mm]\label{Asymptoticproperty1}
   &&\lim\limits_{\overline{T}\rightarrow\infty}q^1(0;A,0,\overline{T})\leq(p-1)\ln\overline{c}+\ln\lambda
  \Leftrightarrow \rho-pK\geq(1-p)\overline{c};
  \ee
  \be\nonumber
   &&\lim\limits_{\overline{T}\rightarrow\infty}q^2(0;A,0,\overline{T})=\left\{
  \begin{array}{ll}
  \ln\lambda+(1-p)\ln\left(\,{1-p\over\rho-pK}\,\right),
  &\rho-pK>0;
  \vspace{2mm} \\
  +\infty,
  &\rho-pK\leq0;
  \end{array}
  \right.
  \\[2mm]\label{Asymptoticproperty2}
   &&\lim\limits_{\overline{T}\rightarrow\infty}q^2(0;A,0,\overline{T})\geq(p-1)\ln\overline{c}+\ln\lambda
  \Leftrightarrow \rho-pK\leq (1-p)\overline{c};
  \\[2mm]\label{Asymptoticproperty3}
   &&\lim\limits_{\overline{T}\rightarrow\infty}q^2(0;A,0,\overline{T})
   \leq(p-1)\ln\underline{c}+\ln\lambda
  \Leftrightarrow \rho-pK\geq (1-p)\underline{c};
  \ee
  \be\nonumber
   &&\lim\limits_{\overline{T}\rightarrow\infty}q^3(0;A,0,\overline{T})=
  \left\{
  \begin{array}{ll}
  \ln\lambda+\ln\left(\,{\underline{c}^p\over \rho+p\underline{c}-pK}\,\right),
  &\rho-pK>-p\underline{c};
  \vspace{2mm} \\
  +\infty,
  &\rho-pK\leq-p\underline{c};
  \end{array}
  \right.
  \\[2mm]\label{Asymptoticproperty4}
   &&\lim\limits_{\overline{T}\rightarrow\infty}q^3(0;A,0,\overline{T})\geq(p-1)\ln\underline{c}+\ln\lambda
  \Leftrightarrow \rho-pK\leq (1-p)\underline{c}.
 \ee

\noindent{\it Proof of Theorem \ref{theorem_consumption}.}

\noindent{\underline{Case (1)}} $0\leq
\underline{c}<\overline{c}<\lambda^{1/(1-p)}$.\smallskip

In this case, $(p-1)\ln\overline{c}+\ln\lambda>0$. Since $q_P(T)=0$,
then, when $t$ is close to $T$,
$q_P(t)<(p-1)\ln\overline{c}+\ln\lambda$ and
$c^*(t)=\widetilde{x}_{c,P}^*(q_P(t))=\overline{c}.$ Thus $q_P(t)$
satisfies ODE (\ref{ODE1}) with $A=1/\lambda$ and $\overline{T}=T$
in the interval $[\,\underline{T},T\,]$, until $\underline{T}=0$ or
$q_P(\underline{T})=(p-1)\ln\overline{c}+\ln\lambda$. \smallskip

\noindent{\bf (1.1)} If $\rho-pK<(1-p)\underline{c}\neq0$, solving
ODE (\ref{ODE1}), we obtain $q_P(t)=q^1(t;1/\lambda,\underline{T},T)$
taking the form of (\ref{equ20}) in the interval
$[\,\underline{T},T\,]$. According to (\ref{Asymptoticproperty1}),
$q^1(0;1/\lambda,0,T)>(p-1)\ln\overline{c}+\ln\lambda$ provided $T$
is large enough. Thus, there exists a positive constant $T_{12}$
such that
$q^1(T_{12};1/\lambda,T_{12},T)=(p-1)\ln\overline{c}+\ln\lambda$,
and $T_{12}$ is given in (\ref{equ23}). Hence, we derive
$q_P(t)=q^1(t;1/\lambda,T_{12},T)$, and $c^*(t)=\overline{c}$ in the
interval $[\,T_{12},T\,]$.

Since
$$
  q^\prime_P(T_{12})=\rho-pf_P(q_P(T_{12}),c^*(T_{12}))-pg(x^*_\pi;x^*_\mu,{x_\sigma^*})=\rho-(1-p)\overline{c}-pK<0,
$$
then, when $t<T_{12}$ and $t$ is close to $T_{12}$,
$(p-1)\ln\overline{c}+\ln\lambda<q_P(t)<(p-1)\ln\underline{c}+\ln\lambda$,
and $c^*(t)=\widehat{c}(t)$ taking the form of (\ref{widehat_c}).
Thus $q_P(t)$ satisfies ODE (\ref{ODE2}) with $A=\overline{c}^{p-1}$
and $\overline{T}=T_{12}$ in the interval
$[\,\underline{T},T_{12}\,]$, until $\underline{T}=0$ or
$q_P(\underline{T})=(p-1)\ln\underline{c}+\ln\lambda$ (where we have
used the fact that the sign of $q_P^\prime(t)$ does not change (cf.
Proposition \ref{proposition}), and
$q_P(t)>(p-1)\ln\overline{c}+\ln\lambda$ for any $t<T_{12}$).

Solving ODE (\ref{ODE2}), we obtain
$q_P(t)=q^2(t;\overline{c}^{p-1},\underline{T},T_{12})$ taking the
form of (\ref{equ21}) in the interval $[\,\underline{T},T_{12}\,]$.
According to (\ref{Asymptoticproperty3}),
$q^2(0;\overline{c}^{p-1},0,T_{12})>(p-1)\ln\underline{c}+\ln\lambda$
provided $T$ is large enough. Thus, there exists a positive constant
$T_{123}$ such that
$q^2(T_{123};\overline{c}^{p-1},T_{123},T_{12})=(p-1)\ln\underline{c}+\ln\lambda$,
and $T_{123}$ is given in (\ref{equ24}). Hence,  we derive
$q_P(t)=q^2(t;\overline{c}^{p-1},T_{123},T_{12})$, and
$c^*(t)=\widehat{c}(t)$ in the interval $[\,T_{123},T_{12}\,]$.

Recalling the fact that the sign of $q_P^\prime(t)$ does not change
(cf. Proposition \ref{proposition}), we deduce that $q_P(t)\geq
(p-1)\ln\underline{c}+\ln\lambda,\,c^*(t)=\underline{c}$, and
$q_P(t)$ satisfies ODE (\ref{ODE3}) with $A=\underline{c}^{p-1}$ and
$\overline{T}=T_{123}$ in the interval $[\,0,T_{123}\,]$. Solving
ODE (\ref{ODE3}), we have
$q_P(t)=q^3(t;\underline{c}^{p-1},0,T_{123})$ as in
(\ref{equ22}).\smallskip

\noindent {\bf (1.2)} If $\rho-pK<(1-p)\underline{c}=0$ or
$(1-p)\underline{c}\leq \rho-pK<(1-p)\overline{c}$, repeating the
same argument as in Case (1.1), we have
$q_P(t)=q^1(t;1/\lambda,T_{12},T)$ as in (\ref{equ20}), and
$c^*=\overline{c}$ in the interval $[\,T_{12},T\,]$, and
$q_P(t)=q^2(t;\overline{c}^{p-1},\underline{T},T_{12})$ as in
(\ref{equ21}) until $\underline{T}=0$ or
$q_P(\underline{T})=(p-1)\ln\underline{c}+\ln\lambda$.

In the case of $\rho-pK\leq (1-p)\underline{c}=0$,
$(p-1)\ln\underline{c}+\ln\lambda=+\infty>q^2(t;\overline{c}^{p-1},0,T_{12})$,
and $\underline{T}=0$. In the other case, since $\rho-pK>0$ and
$\rho-pK-(1-p)\overline{c}<0$, then (\ref{equ21}) implies that
$$
  q^2(t;\overline{c}^{p-1},0,T_{12})<\ln\lambda+(1-p)\ln{1-p\over \rho-pK}
  \leq \ln\lambda+(1-p)\ln{1\over
  \underline{c}}=(p-1)\ln\underline{c}+\ln\lambda,\;\;\forall\;t\in [\,0,T_{12}\,].
$$
Thus, we deduce that $\underline{T}=0$. Therefore,
$q_P(t)=q^2(t;\overline{c}^{p-1},0,T_{12})$  and
$c^*(t)=\widehat{c}(t)$ in the interval $[\,0,T_{12}\,]$.\smallskip

\noindent{\bf (1.3)} If $(1-p)\overline{c}\leq
\rho-pK<\lambda\overline{c}^p-p\overline{c}$, solving ODE
(\ref{ODE1}), we have $q_P(t)=q^1(t;1/\lambda,\underline{T},T)$ as in
(\ref{equ20}) until  $\underline{T}=0$ or
$q_P(\underline{T})=(p-1)\ln\overline{c}+\ln\lambda$. Since
$$
  q^1(t;1/\lambda,0,T)<\ln\lambda+\ln{\overline{c}^p\over \rho+p\overline{c}-pK}
  \leq \ln\lambda+\ln\overline{c}^{p-1}=(p-1)\ln\overline{c}+\ln\lambda,\;\;\forall\;t\in
  [\,0,T\,],
$$
then $q_P(t)=q^1(t;1/\lambda,0,T)$  and $c^*(t)=\overline{c}$ in the
interval $[\,0,T\,]$.\smallskip

\noindent{\bf (1.4)} If $\rho-pK\geq
\lambda\overline{c}^p-p\overline{c}$, solving ODE (\ref{ODE1}), we
derive that  $q_P(t)=q^1(t;1/\lambda,\underline{T},T)$ until
$\underline{T}=0$ or
$q_P(\underline{T})=(p-1)\ln\overline{c}+\ln\lambda$.

Since $\rho+p\overline{c}-pK\geq0$ and
$\rho+p\overline{c}-pK-\lambda\overline{c}^p\geq0$, then
$q^1(t;1/\lambda,0,T)$ is nondecreasing with respect to $t$, thus
for $t\in[\,0,T\,]$, $q^1(t;1/\lambda,0,T)\leq
q^1(T;1/\lambda,0,T)=0<(p-1)\ln\overline{c}+\ln\lambda$. Hence,
$q_P(t)=q^1(t;1/\lambda,0,T)$  and $c^*(t)=\overline{c}$ in the
interval $[\,0,T\,]$. \medskip

\noindent{\underline{Case (2)}} $0\leq
\underline{c}<\overline{c}=\lambda^{1/(1-p)}$. In this case, note
that $(p-1)\ln\overline{c}+\ln\lambda=0$.\smallskip

\noindent{\bf (2.1)} If $\rho-pK<(1-p)\underline{c}\neq0$, since
$$
 q_P^\prime(T-0)=\rho-pf_P(q_P(T),c^*(T))-pg(x^*_\pi;x^*_\mu,{x_\sigma^*})=\rho-(1-p)\overline{c}-pK\leq \rho-pK-(1-p)\underline{c}<0,
$$
then, when $t$ is close to $T$,
$q_P(t)>0=(p-1)\ln\overline{c}+\ln\lambda,\,q_P(t)<(p-1)\ln\underline{c}+\ln\lambda$,
and $c^*(t)=\widehat{c}(t)$. Thus $q_P(t)$ satisfies ODE
(\ref{ODE2}) with $A=1/\lambda$ and $\overline{T}=T$ in the interval
$[\,\underline{T},T\,]$, until $\underline{T}=0$ or
$q_P(\underline{T})=(p-1)\ln\underline{c}+\ln\lambda$.

Solving ODE (\ref{ODE2}), we obtain
$q_P(t)=q^2(t;1/\lambda,\underline{T},T)$ as in (\ref{equ21}) in the
interval $[\,\underline{T},T\,]$. According to
(\ref{Asymptoticproperty3}),
$q^2(0;1/\lambda,0,T)>(p-1)\ln\underline{c}+\ln\lambda$ provided
that $T$ is large enough. Thus, there exists a positive constant
$T_{23}$ such that
$q^2(T_{23};1/\lambda,T_{23},T)=(p-1)\ln\underline{c}+\ln\lambda$,
and $T_{23}$ is given in (\ref{equ25}). Hence,  we derive that
$q_P(t)=q^2(t;1/\lambda,T_{23},T)$, and $c^*(t)=\widehat{c}(t)$ in the
interval $[\,T_{23},T\,]$.

Recalling the fact that the sign of $q_P^\prime(t)$ does not change,
we deduce that in the interval $[\,0,T_{23}\,]$, $q_P(t)\geq
q_P(T_{23})= (p-1)\ln\underline{c}+\ln\lambda,\,c^*(t)=\underline{c}$,
and $q_P(t)$ satisfies ODE (\ref{ODE3}) with $A=\underline{c}^{p-1}$
and $\overline{T}=T_{23}$. Solving ODE (\ref{ODE3}), we have
$q_P(t)=q^3(t;\underline{c}^{p-1},0,T_{23})$ as in (\ref{equ22}) in
the interval $[\,0,T_{23}\,]$.\smallskip

\noindent{\bf (2.2)} If $\rho-pK<(1-p)\underline{c}=0$ or
$(1-p)\underline{c}\leq \rho-pK<(1-p)\overline{c}$, since
$$
 q_P^\prime(T-0)=\rho-pf_P(q_P(T),c^*(T))-pg(x^*_\pi;x^*_\mu,{x_\sigma^*})=\rho-(1-p)\overline{c}-pK<0
$$
still holds, then repeating the similar argument as in Case (2.1),
we deduce that $c^*(t)=\widehat{c}(t)$ and
$q_P(t)=q^2(t;1/\lambda,\underline{T},T)$ in the interval
$[\,\underline{T},T\,]$, until $\underline{T}=0$ or
$q_P(\underline{T})=(p-1)\ln\underline{c}+\ln\lambda$.

In the case of $\rho-pK\leq(1-p)\underline{c}=0$,
$(p-1)\ln\underline{c}+\ln\lambda=+\infty>q^2(t;1/\lambda,0,T)$, and
$\underline{T}=0$. In the other case, since $\rho-pK>0$ and
$(\rho-pK)-(1-p)\lambda^{1/(1-p)}=(\rho-pK)-(1-p)\overline{c}<0$,
then for any $t\in [\,0,T\,]$, we still have \be\label{case2.2}
  q^2(t;1/\lambda,0,T)<\ln\lambda+(1-p)\ln{1-p\over \rho-pK}
  \leq \ln\lambda+(1-p)\ln{1\over
  \underline{c}}=(p-1)\ln\underline{c}+\ln\lambda.
\ee Therefore, $q_P(t)=q^2(t;1/\lambda,0,T)$  and
$c^*(t)=\widehat{c}(t)$ in the interval $[\,0,T\,]$.\smallskip

\noindent{\bf (2.3)} If $\rho-pK\geq(1-p)\overline{c}$. We first
discuss the case when $\rho-pK>(1-p)\overline{c}$. Combining the
following calculation
$$
 q_P^\prime(T-0)=\rho-pf_P(q_P(T),c^*(T))-pg(x^*_\pi;x^*_\mu,{x_\sigma^*})=\rho-(1-p)\overline{c}-pK>0,
$$
and the fact that the sign of $q_P^\prime(t)$ does not change, we
drive that $q_P(t)<q_P(T)=0=(p-1)\ln\overline{c}+\ln\lambda,
c^*(t)=\overline{c}$, and $q_P(t)$ satisfies ODE (\ref{ODE1}) with
$A=1/\lambda$ and $\overline{T}=T$ in the interval $[\,0,T\,]$.
Solving ODE (\ref{ODE1}), we have $q_P(t)=q^1(t;1/\lambda,0,T)$ and
$c^*(t)=\overline{c}$ in the interval $[\,0,T\,]$.

On the other hand, if $\rho-pK=(1-p)\overline{c}$, then
$\rho-pK=\lambda\overline{c}^p-p\overline{c}$, and for
$t\in[\,0,T\,]$, we have $q_P(t)=0$, thus still have
$q_P(t)=q^1(t;1/\lambda,0,T)$.\medskip

\noindent{\underline{Case (3)}}
$0\leq\underline{c}<\lambda^{1/(1-p)}<\overline{c}$.\smallskip

In this case, note that
$(p-1)\ln\overline{c}+\ln\lambda<0<(p-1)\ln\underline{c}+\ln\lambda$.
Since $q_P(T)=0$, then, when $t$ is close to $T$,
$(p-1)\ln\overline{c}+\ln\lambda<q_P(t)<(p-1)\ln\underline{c}+\ln\lambda$,
$c^*(t)=\widehat{c}(t)$ and  $q_P(t)$ satisfies ODE (\ref{ODE2})
with $A=1/\lambda$ and $\overline{T}=T$ in the interval
$[\,\underline{T},T\,]$, until $\underline{T}=0$ or
$q_P(\underline{T})=(p-1)\ln\overline{c}+\ln\lambda$ or
$q_P(\underline{T})=(p-1)\ln\underline{c}+\ln\lambda$.\smallskip

\noindent{\bf (3.1)} If $\rho-pK<(1-p)\underline{c}\neq0$, solving
ODE (\ref{ODE2}), we have $q_P(t)=q^2(t;1/\lambda,\underline{T},T)$
and $c^*(t)=\widehat{c}(t)$ in the interval $[\,\underline{T},T\,]$.

Since
$$
  q_P^\prime(T-0)=\rho-pf_P(q_P(T),c^*(T))-pg(x^*_\pi;x^*_\mu,{x_\sigma^*})
  =\rho-(1-p)\lambda^{1/(1-p)}-pK<\rho-(1-p)\underline{c}-pK<0,
$$
then we deduce $q_P(t)$ is nonincreasing with respect to $t$ from the
fact that the sign of $q_P^\prime(t)$ does not change. Hence,  we have
$q_P(t)>(p-1)\ln\overline{c}+\ln\lambda$ for any $t\in[\,0,T\,]$.
Moreover, (\ref{Asymptoticproperty3}) implies that
$q^2(0;1/\lambda,0,T)>(p-1)\ln\underline{c}+\ln\lambda$ provided
that $T$ is large enough. Thus, there exists a positive constant
$T_{23}$ such that
$q^2(T_{23};1/\lambda,T_{23},T)=(p-1)\ln\underline{c}+\ln\lambda$,
and $T_{23}$ is given in (\ref{equ25}). Hence,  we derive that
$q_P(t)=q^2(t;1/\lambda,T_{23},T)$, and $c^*(t)=\widehat{c}(t)$ in the
interval $[\,T_{23},T\,]$.

Since
$$
  q_P^\prime(T_{23})=\rho-pf_P(q_P(T_{23}),c^*(T_{23}))-pg(x^*_\pi;x^*_\mu,{x_\sigma^*})
  =\rho-(1-p)\underline{c}-pK<0,
$$
then for any $t\in[\,0,T_{23})$, we have
$q_P(t)>(p-1)\ln\underline{c}+\ln\lambda$ , and $q_P(t)$ satisfies ODE
(\ref{ODE3}) with $A=\underline{c}^{p-1}$ and $\overline{T}=T_{23}$
in the interval $[\,0,T_{23}\,]$. Solving ODE (\ref{ODE3}), we
obtain $q_P(t)=q^3(t;\underline{c}^{p-1},0,T_{23})$ and
$c^*(t)=\underline{c}$ in the interval $[\,0,T_{23}\,]$.\smallskip

\noindent{\bf (3.2)} If $\rho-pK<(1-p)\underline{c}=0$ or
$(1-p)\underline{c}\leq\rho-pK<(1-p)\lambda^{1/(1-p)}$, repeating
the similar argument as in case (3.1), we deduce that
$q_P(t)=q^2(t;1/\lambda,\underline{T},T),\,c^*(t)=\widehat{c}(t)$ in
the interval $[\,\underline{T},T\,]$, until  $\underline{T}=0$ or
$q_P(\underline{T})=(p-1)\ln\underline{c}+\ln\lambda$.

For the case of $\rho-pK\leq (1-p)\underline{c}=0$,
$(p-1)\ln\underline{c}+\ln\lambda=+\infty>q^2(t;1/\lambda,0,T)$, and
$\underline{T}=0$. For the other case, since $\rho-pK>0$ and
$\rho-pK-(1-p)\lambda^{1/(1-p)}<0$, then (\ref{case2.2}) still
holds. Therefore, $q_P(t)=q^2(t;1/\lambda,0,T)$ and
$c^*(t)=\widehat{c}(t)$ in the interval $[\,0,T\,]$.\smallskip

\noindent{\bf (3.3)} If $(1-p)\lambda^{1-p}\leq\rho-pK\leq
(1-p)\overline{c}$, solving ODE (\ref{ODE2}), we have
$q_P(t)=q^2(t;1/\lambda,\underline{T},T)$ and $c^*(t)=\widehat{c}(t)$
in the interval $[\,\underline{T},T\,]$. Since $\rho-pK>0$ and
$\rho-pK-(1-p)\lambda^{1-p}\geq 0$, then $q^2(t;1/\lambda,0,T)$ is
nondecreasing and for $t\in[\,0,T\,]$, we have \bee
  (p-1)\ln\underline{c}+\ln\lambda&>&0=q^2(T;1/\lambda,0,T)\geq q^2(t;1/\lambda,0,T)
  \geq(1-p)\ln\left(\,{1-p\over\rho-pK}\,\right)+\ln\lambda
  \\[2mm]
  &\geq&
  (1-p)\ln{1\over\overline{c}}+\ln\lambda=(p-1)\ln\overline{c}+\ln\lambda.
\eee Therefore, $q_P(t)=q^2(t;1/\lambda,0,T)$ and
$c^*(t)=\widehat{c}(t)$ in the interval  $[\,0,T\,]$.\smallskip

\noindent{\bf (3.4)} If $\rho-pK>(1-p)\overline{c}$,  solving ODE
(\ref{ODE2}), we have $q_P(t)=q^2(t;1/\lambda,\underline{T},T)$ and
$c^*(t)=\widehat{c}(t)$ in the interval $[\,\underline{T},T\,]$.

Since
$$
  q_P^\prime(T-0)=\rho-pf_P(q_P(T),c^*(T))-pg(x^*_\pi;x^*_\mu,{x_\sigma^*})
  =\rho-(1-p)\lambda^{1/(1-p)}-pK>\rho-(1-p)\overline{c}-pK>0,
$$
then we deduce $q_P(t)$ is nondecreasing with respect to $t$ from the
fact that the sign of $q_P^\prime(t)$ does not change. Hence,  we have
$q_P(t)<(p-1)\ln\underline{c}+\ln\lambda$ for any $t\in[\,0,T\,]$.
Moreover, (\ref{Asymptoticproperty2}) implies that
$q^2(0;1/\lambda,0,T)<(p-1)\ln\overline{c}+\ln\lambda$ provided that
$T$ is large enough. Thus, there exists a positive constant $T_{21}$
such that
$q^2(T_{21};1/\lambda,T_{21},T)=(p-1)\ln\overline{c}+\ln\lambda$,
and $T_{21}$ is given in (\ref{equ26}). Hence,  we derive that
$q_P(t)=q^2(t;1/\lambda,T_{21},T)$, and $c^*(t)=\widehat{c}(t)$ in the
interval $[\,T_{21},T\,]$.

Since
$$
  q_P^\prime(T_{21})=\rho-pf_P(q_P(T_{21}),c^*(T_{21}))-pg(x^*_\pi;x^*_\mu,{x_\sigma^*})
  =\rho-(1-p)\overline{c}-pK>0,
$$
then for any $t\in[\,0,T_{12})$, we have
$q_P(t)<(p-1)\ln\overline{c}+\ln\lambda$ , and $q_P(t)$ satisfies ODE
(\ref{ODE1}) with $A=\overline{c}^{p-1}$ and $\overline{T}=T_{12}$
in the interval $[\,0,T_{12}\,]$. Solving ODE (\ref{ODE1}), we
obtain $q_P(t)=q^1(t;\overline{c}^{p-1},0,T_{12})$ and
$c^*(t)=\overline{c}$ in the interval $[\,0,T_{12}\,]$.\medskip

\noindent{\underline{Case (4)}}
$\lambda^{1/(1-p)}=\underline{c}<\overline{c}$. In this case, note
that
$(p-1)\ln\overline{c}+\ln\lambda<0=(p-1)\ln\underline{c}+\ln\lambda$.\smallskip

\noindent{\bf (4.1)} If $\rho-pK\leq (1-p)\underline{c}$, we first
consider the case where $\rho-pK<(1-p)\underline{c}$. Combining the
following calculation
$$
 q_P^\prime(T-0)=\rho-pf_P(q_P(T),c^*(T))-pg(x^*_\pi;x^*_\mu,{x_\sigma^*})
 =\rho-(1-p)\underline{c}-pK<0,
$$ and the fact that  the sign of $q_P^\prime(t)$ does not change, we deduce
that $q_P(t)>0=(p-1)\ln\underline{c}+\ln\lambda,
c^*(t)=\underline{c}$, and $q_P(t)$ satisfies ODE (\ref{ODE3}) with
$A=1/\lambda$ and $\overline{T}=T$ in the interval $[\,0,T\,]$.
Solving ODE (\ref{ODE3}), we have $q_P(t)=q^3(t;1/\lambda,0,T)$ and
$c^*(t)=\underline{c}$ in the interval $[\,0,T\,]$.\smallskip

When $\rho-pK=(1-p)\underline{c}$, it is easy to see that for
$t\in[\,0,T\,]$, we have $q_P(t)=0$, and we still have $q_P(t)$ equal to
$q^3(t;1,0,T)$ and $c^*(t)=\underline{c}$ in the interval
$[\,0,T\,]$.\smallskip

\noindent{\bf (4.2)} If $(1-p)\underline{c}<\rho-pK\leq
(1-p)\overline{c}$, since $q(T)=0=(p-1)\ln\underline{c}+\ln\lambda$,
and
$$
 q_P^\prime(T-0)=\rho-pf_P(q_P(T),c^*(T))-pg(x^*_\pi;x^*_\mu,{x_\sigma^*})
 =\rho-(1-p)\underline{c}-pK>0,
$$ then, when $t$ is close to $T$, we have
$(p-1)\ln\overline{c}+\ln\lambda<q_P(t)<(p-1)\ln\underline{c}+\ln\lambda$.
Thus, $q_P(t)$ satisfies ODE (\ref{ODE2}) with $A=1/\lambda$ and
$\overline{T}=T$ in the interval $[\,\underline{T},T\,]$, until
$\underline{T}=0$ or
$q_P(\underline{T})=(p-1)\ln\overline{c}+\ln\lambda$ or
$q_P(\underline{T})=(p-1)\ln\underline{c}+\ln\lambda$. Recalling the
fact that the sign of $q_P^\prime(t)$ does not change, we deduce that
$q_P(t)$ is nondecreasing with respect to $t$. Thus, it is impossible
that $q_P(\underline{T})=(p-1)\ln\underline{c}+\ln\lambda$ for some
$\underline{T}\in[\,0,T)$.

Solving ODE (\ref{ODE2}), we have
$q_P(t)=q^2(t;1/\lambda,\underline{T},T)$ and $c^*(t)=\widehat{c}(t)$
in the interval $[\,\underline{T},T\,]$. Since
$(\rho-pK)-(1-p)\lambda^{1/(1-p)}=(\rho-pK)-(1-p)\underline{c}>0$,
then $q^2(t;1/\lambda,0,T)$ is increasing with respect to $t$, thus
for $t\in[\,0,T)$, we have \bee
 (p-1)\ln\underline{c}+\ln\lambda&=&0=q^2(T;1/\lambda,0,T)>q^2(t;1/\lambda,0,T)
 >(1-p)\ln{1-p\over\rho-pK}+\ln\lambda
 \\[2mm]
 &\geq&(1-p)\ln{1\over\overline{c}}+\ln\lambda=(p-1)\ln\overline{c}+\ln\lambda.
\eee Therefore, $q_P(t)=q^2(t;1/\lambda,0,T)$ and
$c^*(t)=\widehat{c}(t)$ in the interval $[\,0,T\,]$.\smallskip

\noindent{\bf (4.3)} If $\rho-pK>(1-p)\overline{c}$, repeating the
similar argument as in case (4.2), we deduce that
$q_P(t)=q^2(t;1/\lambda,\underline{T},T)$ and $c^*(t)=\widehat{c}(t)$
in the interval $[\,\underline{T},T\,]$, until $\underline{T}=0$ or
$q_P(\underline{T})=(p-1)\ln\overline{c}+\ln\lambda$.

According to (\ref{Asymptoticproperty2}),
$q^2(0;1/\lambda,0,T)<(p-1)\ln\overline{c}+\ln\lambda$ provided that
$T$ is large enough. Thus, there exists a positive constant $T_{21}$
such that
$q^2(T_{21};1/\lambda,T_{21},T)=(p-1)\ln\overline{c}+\ln\lambda$,
and $T_{21}$ is given in (\ref{equ26}). Hence,  we derive that
$q_P(t)=q^2(t;1/\lambda,T_{21},T)$, and $c^*(t)=\widehat{c}(t)$ in the
interval $[\,T_{21},T\,]$.

Combining
$$
 q_P^\prime(T_{21})=\rho-pf_P(q_P(T_{21}),c^*(T_{21}))-pg(x^*_\pi;x^*_\mu,{x_\sigma^*})
 =\rho-pK-(1-p)\overline{c}>0,
$$
and the fact that the sign of $q_P^\prime(t)$ does not change, we
deduce that  in the interval $[\,0,T_{21})$,
$q_P(t)<(p-1)\ln\overline{c}+\ln\lambda,\,c^*(t)=\overline{c}$, and
$q_P(t)$ satisfies ODE (\ref{ODE1}) with $A=\overline{c}^{p-1}$ and
$\overline{T}=T_{21}$. Solving ODE (\ref{ODE1}), we obtain
$q_P(t)=q^1(t;\overline{c}^{p-1},0,T_{21})$ in the interval
$[\,0,T_{21}\,]$.\medskip

\noindent{\underline{Case (5)} }
$\lambda^{1/(1-p)}<\underline{c}<\overline{c}$.\smallskip

Since $q_P(T)=0>(p-1)\ln\underline{c}+\ln\lambda$, then $q_P(t)$
satisfies ODE (\ref{ODE3}) with $A=1/\lambda$ and $\overline{T}=T$
in the interval $[\,\underline{T},T\,]$, until $\underline{T}=0$ or
$q_P(\underline{T})=(p-1)\ln\underline{c}+\ln\lambda$. Solving ODE
(\ref{ODE3}), we obtain $q_P(t)=q^3(t;1/\lambda,\underline{T},T)$ and
$c^*(t)=\underline{c}$ in the interval
$[\,\underline{T},T\,]$.\smallskip

\noindent{\bf (5.1)} If
$\rho-pK\leq\lambda\underline{c}^p-p\underline{c}$, then
$\rho+p\underline{c}-pK-\lambda\underline{c}^p\leq 0$, and
$q^3(t;1/\lambda,0,T)$ is nonincreasing with respect to $t$, thus
for $t\in[\,0,T\,]$, we have $q^3(t;1/\lambda,0,T)\geq
q^3(T;1/\lambda,0,T)=0>(p-1)\ln\underline{c}+\ln\lambda$. Therefore,
$q_P(t)=q^3(t;1/\lambda,0,T)$  and $c^*(t)=\underline{c}$ in the
interval $[\,0,T\,]$.\smallskip

\noindent{\bf (5.2)} If
$\lambda\underline{c}^p-p\underline{c}<\rho-pK\leq(1-p)\underline{c}$,
then $\rho+p\underline{c}-pK-\lambda\underline{c}^p>0$, and
$$
  q^3(t;1/\lambda,0,T)\geq \ln\lambda+\ln{\underline{c}^p\over \rho+p\underline{c}-pK}
  \geq
  \ln{\underline{c}^p\over\underline{c}}+\ln\lambda=(p-1)\ln\underline{c}+\ln\lambda,\;\;\forall\;
  t\in[\,0,T\,].
$$
Therefore we still have $q_P(t)=q^3(t;1/\lambda,0,T)$ and
$c^*(t)=\underline{c}$ in the interval $[0,T]$.\smallskip

\noindent{\bf (5.3)} If
$(1-p)\underline{c}<\rho-pK\leq(1-p)\overline{c}$, then
(\ref{Asymptoticproperty4}) implies that
$q^3(0;1/\lambda,0,T)<(p-1)\ln\underline{c}+\ln\lambda$ provided
that $T$ is large enough. Thus, there exists a positive constant
$T_{32}$ such that
$q^3(T_{32};1/\lambda,T_{32},T)=(p-1)\ln\underline{c}+\ln\lambda$,
and $T_{32}$ is given in (\ref{equ27}). Hence,  we derive
$q_P(t)=q^3(t;1/\lambda,T_{32},T)$, and $c^*(t)=\underline{c}$ in the
interval $[\,T_{32},T\,]$.

Since
$$
  q_P^\prime(T_{32})=\rho-pf_P(q_P(T_{32}),c^*(T_{32}))-pg(x^*_\pi;x^*_\mu,{x_\sigma^*})=\rho-(1-p)\underline{c}-pK>0,
$$
then, when $t<T_{32}$ and $t$ is close to $T_{32}$, we have
$q_P(t)<(p-1)\ln\underline{c}+\ln\lambda$ and
$q_P(t)>(p-1)\ln\overline{c}+\ln\lambda$, and $q_P(t)$ is
nondecreasing with respect to $t$, and $q_P(t)$ satisfies ODE
(\ref{ODE2}) with $A=\underline{c}^{p-1}$ and $\overline{T}=T_{32}$
in the interval $[\,\underline{T},\,T_{32}\,]$, until
$\underline{T}=0$ or
$q_P(\underline{T})=(p-1)\ln\overline{c}+\ln\lambda$. Solving ODE
(\ref{ODE2}), we obtain
$q_P(t)=q^2(t;\underline{c}^{p-1},\underline{T},T_{32})$ and
$c^*(t)=\widehat{c}(t)$ in the interval
$[\,\underline{T},\,T_{32}\,]$.

Since $\rho-pK-(1-p)\lambda^{1/(1-p)}>\rho-pK-(1-p)\underline{c}>0$,
then in this case $q^2(t;\underline{c}^{p-1},0,T_{32})$ is
increasing with respect to $t$, thus for $t\in[\,0,T_{32})$, we have
\bee
  (p-1)\ln\underline{c}+\ln\lambda&=&q^2(T_{32};\underline{c}^{p-1},0,T_{32})
  > q^2(t;\underline{c}^{p-1},0,T_{32})>\ln\lambda+(1-p)\ln{1-p\over\rho-pK}
  \\[2mm]
  &\geq&
  (1-p)\ln{1\over\overline{c}}+\ln\lambda=(p-1)\ln\overline{c}+\ln\lambda.
\eee Therefore, $q_P(t)=q^2(t;\underline{c}^{p-1},0,T_{32})$ and
$c^*(t)=\widehat{c}(t)$ in the interval
$[\,0,\,T_{32}\,]$.\smallskip

\noindent{\bf (5.4)} If $\rho-pK>(1-p)\overline{c}$, repeating the
similar argument as in case (5.3), we deduce that
$q_P(t)=q^3(T_{32};1/\lambda,T_{32},T)$, and $c^*(t)=\underline{c}$ in
the interval $[\,T_{32},T\,]$, and
$q_P(t)=q^2(t;\underline{c}^{p-1},\underline{T},T_{32})$ and
$c^*(t)=\widehat{c}(t)$ in the interval
$[\,\underline{T},\,T_{32}\,]$, until  $\underline{T}=0$ or
$q_P(\underline{T})=(p-1)\ln\overline{c}+\ln\lambda$.

According to (\ref{Asymptoticproperty2}),
$q^2(0;\underline{c}^{p-1},0,T_{32})<(p-1)\ln\overline{c}+\ln\lambda$
provided  that $T$ is large enough. Thus, there exists a positive
constant $T_{321}$ such that
$q^2(T_{321};\underline{c}^{p-1},T_{321},T_{32})=(p-1)\ln\overline{c}+\ln\lambda$,
and $T_{321}$ is given in (\ref{equ28}). Hence,  we derive that
$q_P(t)=q^2(t;\underline{c}^{p-1},T_{321},T_{32})$, and
$c^*(t)=\widehat{c}(t)$ in the interval $[\,T_{321},T_{32}\,]$.

Combining
$$
  q_P^\prime(T_{321})=\rho-pf_P(q_P(T_{321}),c^*(T_{321}))-pg(x^*_\pi;x^*_\mu,{x_\sigma^*})=\rho-(1-p)\overline{c}-pK>0,
$$
and the fact that the sign of $q_P^\prime(t)$ does not change sign, we
deduce  that $q_P(t)$ is nondecreasing with respect to $t$, and
$q_P(t)<(p-1)\ln\overline{c}+\ln\lambda$ for any $t\in[\,0,T_{321})$.
Thus, $c^*(t)=\overline{c}$ and $q_P(t)$ satisfies ODE (\ref{ODE1}) in
the interval $t\in[\,0,T_{321})$. Solving ODE (\ref{ODE1}), we
obtain $q_P(t)=q^1(t;\overline{c}^{p-1},0,T_{321})$ in the interval
$t\in[\,0,T_{321}\,]$. $\hfill{} \Box$


\end{document}